\documentclass[preprint,review,12pt]{elsarticle}

\usepackage{amssymb}
\usepackage{amsmath,amsthm,amsfonts}
\usepackage{algorithm, float}
\usepackage{algpseudocode}
\usepackage{comment}
\usepackage{hyperref}
\usepackage{xurl}
\usepackage{verbatim}
\usepackage{graphicx}
\usepackage{wrapfig}
\usepackage{array}
\usepackage[caption=false,font=normalsize,labelfont=sf,textfont=sf]{subfig}
\usepackage{textcomp}
\usepackage{stfloats}
\usepackage{url}
\usepackage{xcolor}
\graphicspath{ {./} }

\usepackage{datatool}
\usepackage{etoolbox}

\DTLnewdb{bibnotes}

\def\bibnote#1#2{%
  \DTLnewrow{bibnotes}
  \DTLnewdbentry{bibnotes}{mylabel}{#1}
  \DTLnewdbentry{bibnotes}{mynote}{#2}
}

\makeatletter
\patchcmd{\@lbibitem}%
  {\item[\hfil\NAT@anchor{#2}{\NAT@num}]}%
  {%
    \item[\hfil\NAT@anchor{#2}{\NAT@num}]%
    \DTLforeach[\DTLiseq{\mylabel}{#2}]{bibnotes}{\mylabel=mylabel,\mynote=mynote}{\textit{\mynote}}
  }{}{\message{^^JPatching failed^^J}}%

\makeatother

\newcommand{\R}{\mathbb{R}}
\newcommand{\norm}[1]{\left\lVert#1\right\rVert}

\newcommand*{\horzbar}{\rule[.5ex]{2.5ex}{0.5pt}}
\DeclareMathOperator{\Tr}{Tr}

\newtheorem{definition}{Definition}
\newtheorem{theorem}{Theorem}
\newtheorem{property}{Proposition}

\begin{document}

\begin{frontmatter}



\title{Privacy-Preserving Matrix Factorization for Recommendation Systems using Gaussian Mechanism}


\author[inst1]{Sohan Salahuddin Mugdho\corref{cor1}}
\ead{1606099@eee.buet.ac.bd}
%\fntext[fn1]{Student}

\author[inst1]{Hafiz Imtiaz\corref{cor1}}
\ead{hafizimtiaz@eee.buet.ac.bd}
%\fntext[fn2]{Associate Professor}

\affiliation[inst1]{organization={Department of Electrical and Electronic Engineering, Bangladesh University of Engineering and Technology},%Department and Organization
            city={Dhaka},
            postcode={1205}, 
            country={Bangladesh}}

\cortext[cor1]{Corresponding authors}

\begin{abstract}
Building a recommendation system involves analyzing user data, which can potentially leak sensitive information about users. Anonymizing user data is often not sufficient for preserving user privacy. Motivated by this, we propose a privacy-preserving recommendation system based on the differential privacy framework and matrix factorization, which is one of the most popular algorithms for recommendation systems. As differential privacy is a powerful and robust mathematical framework for designing privacy-preserving machine learning algorithms, it is possible to prevent adversaries from extracting sensitive user information even if the adversary possesses their publicly available (auxiliary) information. We implement differential privacy via the Gaussian mechanism in the form of output perturbation and release user profiles that satisfy privacy definitions. We employ R\'enyi Differential Privacy for a tight characterization of the overall privacy loss. We perform extensive experiments on real data to demonstrate that our proposed algorithm can offer excellent utility for some parameter choices, while guaranteeing strict privacy.
\end{abstract}

\end{frontmatter}


\section{Introduction}
\label{sec:Intro}
Recommendation systems are integral parts of the modern lifestyle. Whenever we take services of streaming platforms, such as Netflix, YouTube, Spotify, or any shopping platforms, we want recommendations for movies, songs, or items that we might like. Recommendation systems provide such services as they seek to infer how a user might rate a movie or an item if that user had watched or used it. Based on that predicted rating, the recommender system decides whether to recommend the movie or the item to the user. Recommendation systems are mathematically formulated as a matrix completion problem, which can be solved using several methods. One of such approaches is known as Collaborative Filtering (CF). Refer to the survey for further elaboration on collaborative filtering \cite{su2009survey}. Among the many collaborative filtering algorithms for building recommendation systems, Matrix Factorization (MF) \cite{koren2009matrix}, \cite{candes2009exact} is by far one of the most popular and successful. For example, a matrix factorization based collaborative filtering algorithm
won the ‘Netflix Prize’ competition \cite{bennett2007netflix}. It was an open competition to find the best algorithm for predicting the ratings of user-movie pairs. Unfortunately, this convenient service comes with caveats. Training a recommender system involves analyzing user data. It has been shown that simply anonymizing the data is inadequate to guarantee user privacy. Early studies \cite{aimeur2008alambic}, \cite{calandrino2011you}, \cite{mcsherry2009differentially} show that user data or preferences that may be considered insensitive such as movie ratings, can be exploited to deduce sensitive information such as the user’s real identity, medical conditions, or political inclinations. In fact, the anonymized medical data of a governor of Massachusetts were matched with his publicly available voter registration records to identify his medical records \cite{dwork2014algorithmic} successfully. Another infamous example is the identification of users from their movie rating data. These anonymized data were published in the ‘Netflix Prize’ competition \cite{bennett2007netflix}. However, it was validated that some of the anonymous users were identified by using their publicly available Internet Movie Database (IMDb) rating data \cite{narayanan2006break}. Consequently, developing privacy-preserving recommendation systems has become an active research field in recent years. The first study of a privacy-preserving matrix factorization algorithm was based on a cryptographic technique known as garbled circuits \cite{nikolaenko2013privacy}. Unfortunately, even after utilizing the sparsity characteristic of rating matrices, its computational expense was too much to use in practice. More recently, Differential Privacy \cite{dwork2006calibrating}, has become very popular as a mathematical framework for designing privacy-preserving algorithms due to being lightweight and mathematically rigorous. There are two major variants of differential privacy (DP), $\epsilon$-DP, which is known as pure differential privacy, and $(\epsilon, \delta)$-DP, which is known as approximate differential privacy \cite{dwork2014algorithmic}. $\epsilon$-DP can be realized using the Laplace and Exponential mechanisms; see details in \cite{machanavajjhala2011personalized}, \cite{dwork2006calibrating}, \cite{mcsherry2007mechanism}. It was shown in \cite{machanavajjhala2011personalized} that the Laplace mechanism could achieve almost identical accuracy as the Exponential mechanism.
\par
\noindent\textbf{Related works.} There exist research works that incorporate differential privacy into machine learning algorithms, including recommender systems. The implementation of differential privacy in Geometric Recommendation Algorithms is explored in \cite{mcsherry2009differentially}. Additionally, the incorporation of differential privacy is investigated in \cite{machanavajjhala2011personalized}, which is based on a graph that describes the connections between people and items. Privacy is also introduced in the ‘Slope One algorithm’, which is based on collaborative filtering \cite{park2011differentially}. Differentially private matrix factorization is discussed in \cite{chaudhuri2011differentially}, \cite{chaudhuri2008privacy}, whereas differential privacy is integrated with matrix factorization based collaborative filtering in \cite{liu2015fast}, \cite{hua2015differentially}. While $\epsilon$-DP is stricter in preserving privacy, it often provides worse utility than the approximate differential privacy \cite{dwork2014algorithmic}. Both $\epsilon$-DP and $(\epsilon, \delta)$-DP in recommender systems are studied in \cite{mcsherry2009differentially}. For scenarios with decentralized data, improved algorithms for differentially private matrix and tensor factorization are proposed in \cite{imtiaz2018distributed}. The notion of preserving privacy in distributed optimization is investigated in \cite{han2016differentially}. In such problems, the constraints can contain information related to users that can be exploited. One of the algorithms often used in signal processing is principal component analysis (PCA). The ‘Analyze Gauss’ algorithm releases differentially private PCA by computing private singular subspace via the Gaussian mechanism \cite{dwork2014analyze}. A nearly-linear time algorithm for computing differentially private PCA that utilizes a noisy version of the power method is proposed in \cite{hardt2014noisy}. A differentially private PCA algorithm named Symmetric Noise (SN) is proposed in \cite{imtiaz2016symmetric}, which proposes to use Wishart Noise to ensure better utility for the same level of privacy as the Analyze Gauss. Various mechanisms and implementation methods of differential privacy in different machine learning algorithms are discussed in \cite{sarwate2013signal}.
\par
\noindent\textbf{Our Contributions.}
In this paper, we propose a recommendation system based matrix factorization satisfying approximate differential privacy via the Gaussian mechanism \cite{dwork2014algorithmic}. We show that our algorithm provides an $(\epsilon_\mathrm{opt}, \delta_r)$ differentially private user profile matrix containing the inferred user profile vectors. $\epsilon_\mathrm{opt}$, and $\delta_r$ the privacy parameters that are defined later in Theorem \ref{algorithmprivacy}. The user profile matrix can be outsourced to build recommendation systems that use the released user profile vectors and appropriate movie or item profile vectors. The recommender can compute a simple inner product for predicting movies or items that are more likely to be preferred by a user and recommend those items to that user. We demonstrate our algorithm's potency and other major characteristics with varying privacy levels and other key parameters on three real datasets. We also show that our algorithm's utility can closely match that of the non-private algorithm for certain parameter choices. An interesting problem that comes with $(\epsilon, \delta)$-DP is mentioned in \cite{mcsherry2009differentially}. For any specific value of the parameter $\sigma$ of the Gaussian mechanism, there can be infinitely many combinations of $(\epsilon, \delta)$ pairs. Only the values of $\sigma$ are focused on, and no specific $(\epsilon, \delta)$ pairs are derived in \cite{mcsherry2009differentially}. In this paper, we provide a tighter characterization of the overall privacy of a multi-step algorithm to find the best $(\epsilon, \delta)$ pair using R\'enyi Differential Privacy (RDP) \cite{mironov2017renyi}.
\par
\noindent\textbf{Notations.} Matrices are represented with upper case letters, e.g. $X$ or $\Theta$. We represent matrix elements with corresponding lower case letters accompanied by row, column indices, e.g. $x_{ij}$. Vectors are denoted with lower case letters, with the corresponding row indices as subscript, e.g. $x_i$, the $i$-th row vector of matrix $X$. We denote algorithms, functions, and sets with calligraphic scripts, e.g. $\mathcal{A}$. Inner product and Hadamard product are denoted with the operators $\langle \cdot \rangle$ and $\odot$, respectively. $\norm{\cdot}_F$ denotes the Frobenius norm of a matrix or a vector, $\norm{\cdot}_1$ denotes the $\mathcal{L}_1$ norm and $\norm{\cdot}_2$ denotes the Euclidean (or $\mathcal{L}_2$) norm. The gradient of a function $\mathcal{F}$ with respect to a matrix $X$ is denoted as $\nabla_X \mathcal{F}$. $\Tr(\cdot)$ denotes the trace operation and $\Delta_2(\cdot)$ denotes $\mathcal{L}_2$ sensitivity to be defined later in Definition \ref{definition3}.

\section{Formulation of the Matrix Factorization Problem}
\label{sec:MF}
We review collaborative filtering according to \cite{hua2015differentially}, where $n_m$ movies are to be rated by $n_u$ users. We denote $V \in \R^{n_m \times n_u}$ as the full rating matrix, where $v_{ij}$ is the rating of movie or item $i$ by user $j$. Note that, as users rate only a few movies, the rating matrix $V$ is typically an extremely sparse matrix. We also denote an observation mask $R \in \R^{n_m \times n_u}$, where
\begin{equation*}
r_{ij} = 
\begin{cases}
\text{1,} &\text{if user $j$ has rated movie $i$},\\
\text{0,} &\text{otherwise.}
\end{cases}
\end{equation*}
Recommendation systems seek to estimate how a user would have rated a movie if they had seen it based on the prediction of the missing entries in $V$. One of the most prominent strategies for tackling this type of problem is matrix factorization (MF). In MF, the rating matrix $V$ is approximated by a low rank factorization, \cite{liu2015fast}, i.e., 
\begin{equation}
\label{eq:1}
\text{$V \approx X\Theta^T$, where $X \in \R^{n_m \times n}$, and $\Theta \in \R^ {n_u \times n}$.}
\end{equation}
Each rating is calculated as,
\begin{equation}
\label{eq:2}
v_{ij} \approx \langle x_i,\theta_j\rangle.
\end{equation}
Here, $X$ and $\Theta$ are profile matrices containing the inferred profile vectors \cite{hua2015differentially} of movies and users, respectively. Profile vectors $x_i \in \R^{1\times n}$ and $\theta_j \in \R^{1\times n}$ are inferred for characterizing movie $i$ and user $j$, respectively. Both vectors are row-vectors of dimension $n$, where typically $n \ll n_m \hspace{4pt} and \hspace{4pt} n \ll n_u$. We sometimes denote $X\Theta^T$ as $\hat{V}$ in this paper. The solution then comes from solving the following Regularized Least Squares minimization \cite{hua2015differentially}:
\begin{equation}
\label{cost}
\underset{X,\Theta}{\min}\hspace{2pt}\mathcal{C}(X,\Theta) = 
\underset{X,\Theta}{\min}\hspace{2pt}\frac{1}{2}\norm{\hat{V}-V}^2_F + \frac{1}{2}\lambda\left(\norm{X}^2_F + \norm{\Theta}^2_F\right),
\end{equation}
where $\lambda$ is a positive scalar, also known as the regularization parameter. To find the solution for \eqref{cost}, we perform alternating gradient descent, that is we perform gradient descent to update $X$ and $\Theta$ alternatingly. 
\par
Our objective is to compute and release the user profile matrix $\Theta$ which can be used to predict ratings for the users using third party item profile matrices. Because $\Theta$ needs to be derived using sensitive user data (in our case, ratings), we propose a method for ensuring privacy. More specifically, we use differentially private estimates of the gradient to update, and compute the differentially private user profile matrix $\Theta$ such that it is a close approximation of the true (or, non-privacy-preserving) $\Theta$.
\section{Preliminaries of Privacy}
\label{sec:PrePrivacy}
\subsection{Differential Privacy}
\label{subsec:DP}
Many signal processing and machine learning algorithms operate on private or sensitive data. As mentioned before, the results of such algorithms might reveal personal information about people in the dataset. Differential privacy is a robust cryptographically motivated framework for preserving user privacy \cite{imtiaz2018distributed}, which is introduced in \cite{dwork2006calibrating}, and has quickly become a popular tool for developing privacy-preserving algorithms due to being lightweight and mathematically rigorous \cite{han2016differentially}. Differential privacy provides aggregate information about a dataset without revealing any further information (that can be related to already public/available data) about the individuals that make up the dataset. It essentially ensures that, whether or not an individual's data is available, the result of a differentially private method seen by an adversary be roughly the same. Formally, differential privacy is defined as below:
\begin{definition}[Pure Differential Privacy \cite{sarwate2013signal}]
\label{definition1}
An algorithm $\mathcal{A}(\cdot)$ taking values in a set $\mathcal{T}$ provides $\epsilon$-differential privacy if $\mathbb{P}(\mathcal{A}(\mathcal{D}) \in \mathcal{S}) \le \exp(\epsilon)\cdot\mathbb{P}(\mathcal{A}(\mathcal{D}^\prime) \in \mathcal{S})$, for all measurable $\mathcal{S} \subseteq \mathcal{T}$ and all data sets $\mathcal{D}$ and $\mathcal{D^\prime}$ differing in a single entry.
\end{definition}
\begin{definition}[Approximate Differential Privacy  \cite{sarwate2013signal}]
\label{definition2}
An algorithm $\mathcal{A}(\cdot)$ taking values in a set $\mathcal{T}$ provides $(\epsilon, \delta)$ differential privacy if $\mathbb{P}(\mathcal{A}(\mathcal{D}) \in \mathcal{S}) \le \exp(\epsilon)\cdot\mathbb{P}(\mathcal{A}(\mathcal{D}^\prime) \in \mathcal{S}) + \delta$, for all measurable $\mathcal{S} \subseteq \mathcal{T}$ and all data sets $\mathcal{D}$ and $\mathcal{D^\prime}$ differing in a single entry.
\end{definition}
Here, $\epsilon$ and $\delta$ are privacy parameters. $\epsilon$ is interpreted as the privacy risk, and $\delta$ is interpreted as the probability that the privacy guarantee fails~\cite{imtiaz2018distributed}. We want both quantities to be as small as possible to ensure better privacy. The privacy guarantee of definition \ref{definition2} is interpreted as being $\epsilon$-differentially private \emph{except with probability $\delta$} \cite{mironov2017renyi}. That is, it offers a slightly weaker privacy guarantee for $\delta > 0$, and reduces to definition \ref{definition1}, pure $\epsilon$-differentially private for $\delta = 0$.
\par
There are several methodologies for ensuring differential privacy for an algorithm. Output perturbation \cite{chaudhuri2011differentially}, \cite{sarwate2013signal} is one of such methods. In this paper, we utilize output perturbation using the Gaussian Mechanism \cite{dwork2006calibrating}, \cite{dwork2014algorithmic}, which satisfies $(\epsilon, \delta)$-DP. Below, we review the definition and theorem of the Gaussian mechanism from \cite{dwork2014algorithmic}.

\begin{definition}[Gaussian Mechanism \cite{dwork2014algorithmic}]
\label{definition3}
Let $\mathcal{F} : \mathbb{N}^{|\mathcal{X}|}
\rightarrow \mathbb{R}^d$ be an arbitrary $d$-dimensional function, and define its $\mathcal{L}_2$ sensitivity be: $\Delta_2\mathcal{F} = \underset{\text{adjacent x,y}}{\max} \norm{\mathcal{F}(x) - \mathcal{F}(y)}_2$. The \emph{Gaussian Mechanism} with parameter $\sigma$ adds noise scaled to $\mathcal{N}(0,\sigma^2)$ to each of the $d$ components of the output.
\end{definition}
\begin{theorem}[Privacy of the Gaussian Mechanism \cite{dwork2014algorithmic}]
\label{theoremGuassian}
Given arbitrary $\epsilon \in (0, 1)$, the Gaussian Mechanism with parameter $\sigma \ge \frac{\Delta_2\mathcal{F}}{\epsilon}\sqrt{2\log\frac{1.25}{\delta}}$ is ($\epsilon, \delta$)-differentially private.
\end{theorem}
\noindent We refer the reader to \cite{dwork2014algorithmic} for the detailed proof of Theorem \ref{theoremGuassian}.

\subsection{R\'enyi Differential Privacy}
\label{subsec:RDP}
Determining the overall privacy loss of a multi-step algorithm is challenging. For ($\epsilon, \delta$) differential privacy, the advanced composition theorem \cite{dwork2014algorithmic} can be loose \cite{mironov2017renyi}. The main reason is that for a given parameter $\sigma$, there can be an endless number of combinations of $\epsilon$ and $\delta$. R\'enyi Differential Privacy (RDP) proposes a much more straightforward composition rule that has been demonstrated to be reliable \cite{mironov2017renyi}. We first review some important propositions of the RDP below:

\begin{property}[From RDP to Differential privacy \cite{mironov2017renyi}]
\label{prop1}
If $\mathcal{A}$ is an ($\alpha, \epsilon_r$)-RDP mechanism, then it also satisfies ($\epsilon_r + \frac{\log\frac{1}{\delta_r}}{\alpha-1},\delta_r$)-differential privacy for any $0 < \delta_r < 1$ and $\alpha > 1$. 
\end{property}
\begin{property}[Composition of RDP \cite{mironov2017renyi}]
\label{prop2}
Let $\mathcal{A} : \mathbb{D} \rightarrow \mathbb{T}_1$ be $(\alpha, \epsilon_{r1})$-RDP and $\mathcal{B} : \mathbb{T}_1 \times \mathbb{D} \rightarrow \mathbb{T}_2$ be $(\alpha, \epsilon_{r2})$-RDP, then the mechanism defined as (X,Y), where $X \sim \mathcal{A}(D)$ and $Y \sim \mathcal{B}(X, D)$, satisfies $(\alpha, \epsilon_{r1} + \epsilon_{r2})$-RDP.
\end{property}
\begin{property}[RDP and Gaussian Mechanism \cite{mironov2017renyi}] 
\label{prop3}
If $\mathcal{A}$ has $\mathcal{L}_2$ sensitivity 1, then the Gaussian mechanism $G_{\sigma}\mathcal{A}(D) = \mathcal{A}(D) + E, where E \sim \mathcal{N}(0,\sigma^2)$ satisfies $(\alpha, \frac{\alpha}{2\sigma^2})$-RDP. Additionally, a composition of J Gaussian mechanisms satisfies $(\alpha,\frac{\alpha J}{2\sigma^2})$-RDP
\end{property}
\section{Proposed Differentially Private Matrix Factorization}
\label{sec:DPMF}
\noindent\textbf{Trust Model.} We assume that the recommender is the data curator (that is, the recommender holds the user ratings) and is trustworthy (that is, the recommender would not reveal the ratings of particular users, and will not collude with an adversary) \cite{han2016differentially}. The curator needs to provide aggregate information without divulging user-specific information, as the queries may come from an adversary who intends to learn sensitive information about a user. We refer the reader to surveys \cite{sarwate2013signal}, \cite{dwork2008differential}, and a textbook on this topic \cite{dwork2014algorithmic} for further details.
\par
\noindent\textbf{Algorithm Formulation.} Recall that for computing the user profile matrix from the given ratings, we need to solve the optimization problem given in \eqref{cost} using an alternating least squares approach. To that end, we use gradient descent with $X$ and $\Theta$ as optimization variables. We now analytically find the gradients of the loss function $\mathcal{C}(X, \Theta)$ with respect to $X$ and $\Theta$. First, we find the gradient of $\mathcal{C}(X,\Theta)$ with respect to $X$ as:
\begin{equation}
\begin{aligned}
\label{gradx}
\nabla_X \mathcal{C}(X,\Theta)
&= \nabla_X\frac{1}{2} \bigg[\norm{\hat{V}-V}^2_F +\lambda\left(\norm{X}^2_F+\norm{\Theta}^2_F\right)\bigg] \\
&= \nabla_X\frac{1}{2} \bigg[\Tr\left(X\Theta^T-V\right)^T\left(X\Theta^T-V\right) \\
& \hspace{12pt} +\lambda \Tr\left(X^TX\right) + \lambda \Tr\left(\Theta^T\Theta\right) \bigg]\\
&=\nabla_X\frac{1}{2}\Tr\left(\Theta X^TX\Theta^T\right) \\
& \hspace{12pt} - \nabla_X\frac{1}{2} \times 2\Tr\left(\Theta X^TV\right) +\nabla_X\frac{1}{2}\Tr\left(V^TV\right) \\
& \hspace{12pt} +\nabla_X\frac{1}{2}\Tr\left(X^TX\right) + \nabla_X\frac{1}{2}\Tr\left(\Theta^T\Theta\right)\\
&=\nabla_X\frac{1}{2}\Tr\left(X^TX\Theta^T\Theta\right) - \nabla_X\Tr\left(X\Theta^TV^T\right) + \lambda X\\
&=X\Theta^T\Theta - V\Theta + \lambda X\\
\text{Therefore, }\nabla_X \mathcal{C}(X,\Theta)
&=\left(\hat{V} - V\right)\Theta + \lambda X.
\end{aligned}
\end{equation}
We use Properties 541, and 14 from \cite{petersen2008matrix} in line 2 of the derivation. We additionally use Properties 13, and 35 in line 3; Properties 16, 13, 115, and 33 in line 4; and, Properties 113, and 100 in line 5.

Next, we find the gradient of $\mathcal{C}(X,\Theta)$ with respect to $\Theta$ as:
\begin{equation}
\begin{aligned}
\label{gradtheta}
\nabla_\Theta \mathcal{C}(X,\Theta) 
&= \nabla_\Theta\frac{1}{2} \bigg[\norm{\hat{V}-V}^2_F +\lambda\left(\norm{X}^2_F+\norm{\Theta}^2_F\right)\bigg]\\
&= \nabla_\Theta\frac{1}{2} \bigg[\Tr\left(X\Theta^T-V\right)^T\left(X\Theta^T-V\right) \\
& \hspace{12pt} +\lambda \Tr\left(X^TX\right) + \lambda \Tr\left(\Theta^T\Theta\right) \bigg]\\
&=\nabla_\Theta\frac{1}{2}\Tr\left(\Theta X^TX\Theta^T\right) \\
& \hspace{12pt} - \nabla_\Theta\frac{1}{2} \times 2\Tr\left(\Theta X^TV\right) + \nabla_\Theta\frac{1}{2}\Tr\left(V^TV\right)\\
& \hspace{12pt} +\nabla_\Theta\frac{1}{2}\Tr\left(X^TX\right) + \nabla_\Theta\frac{1}{2}\Tr\left(\Theta^T\Theta\right)\\
&=\nabla_\Theta\frac{1}{2}\Tr\left(\Theta^T\Theta X^TX\right) - \nabla_\Theta \Tr\left(\Theta X^TV\right) + \lambda X\\
&=\Theta X^TX - V^TX + \lambda \Theta \hspace{4pt}\\
\text{Therefore, }\nabla_\Theta \mathcal{C}(X,\Theta)
&=\left(\hat{V} - V\right)^TX + \lambda \Theta.
\end{aligned}
\end{equation}
In this derivation, we use the same properties from \cite{petersen2008matrix} as mentioned before. Additionally, we use Properties 4, and 5 from \cite{petersen2008matrix} in line 6. We simply denote $\nabla_X \mathcal{C}(X,\Theta)$ and $\nabla_\Theta \mathcal{C}(X,\Theta)$ as $\nabla_X$ and $\nabla_\Theta$, respectively.
\par
Note that, the gradients (as shown in \eqref{gradx} and \eqref{gradtheta}), depend on private data (ratings in our case). As mentioned before, we aim to publish a differentially private user profile matrix that prevents an adversary from inferring individual user information. Therefore, we need to estimate $\nabla_\Theta$ at each iteration satisfying differential privacy. To that end, we employ the Gaussian mechanism~\cite{dwork2014algorithmic}.
\par
\noindent \textbf{Proposed Algorithm.} We initialize $X \in \mathbb{R}^{n_m \times n}$, and  $\Theta \in \mathbb{R}^{n_u \times n}$ with entries drawn i.i.d. from $\mathcal{N}(0, 1)$ and ensure that each profile vector has $\mathcal{L}_2$ norm 1 (that is, replace $x_i = \frac{x_i}{\|x_i\|_2}$ for $i\in \{1, \ldots, n_m\}$, and $\theta_j = \frac{\theta_j}{\|\theta_j\|_2}$ for $j\in \{1, \ldots, n_u\}$). We iterate the following steps until the algorithm converges: we compute the inferred rating matrix as $\hat{V} \gets X\Theta^T \odot R$. Next, we compute $\nabla_X$ according to \eqref{gradx}. We use the norm clipped $\Theta$ in \eqref{gradx} by computing, $\theta_j = \theta_j / \max\left(1, \frac{\norm{\theta_j}_2}{C}\right)$. Here, $C$ is the norm clipping parameter to ensure that $\|\theta_j\|_2 \leq C\ \forall j \in \{1, \ldots, n_u\}$. Similarly, we compute $\nabla_\Theta$ according to \eqref{gradtheta}, and as before, we use the norm clipped $X$ in \eqref{gradtheta} by computing, $x_i = x_i / \max\left(1, \frac{\norm{x_i}_2}{C}\right)$, which ensures $\|x_i\|_2 \leq C\ \forall i\in \{1, \ldots, n_m\}$. This operation is important because, we need to impose a bound on $\norm{x_i}_2$ in order to compute the sensitivity $\Delta_2\left(\nabla_\Theta\right)$ using \eqref{sensitivity}. More details are discussed in the proof of Theorem \ref{algorithmprivacy}. Next, we sample $\eta \in \mathbb{R}^{n_u \times n} \sim$ i.i.d $\mathcal{N}(0,\sigma^2)$, where $\sigma = \frac{\tau C}{\epsilon_i}\sqrt{2\log\frac{1.25}{\delta}}$. Here, $\tau$ is the range of ratings, $C$ is the norm clipping parameter, and $(\epsilon_i,\delta)$ are privacy parameters. We compute $\hat{\nabla}_\Theta = \nabla_\Theta + \eta$ in accordance with Definition \ref{definition3}. Finally, we update $X$, and $\Theta$ with $\nabla_X$, and $\hat{\nabla}_\Theta$ respectively, using step-size $\mu$. Upon completion of the iteration, we release differentially private user profile matrix $\Theta$. We summarize this proposed scheme in Algorithm \ref{algorithm}.
\begin{algorithm}[t]
\caption{Matrix Factorization Satisfying Differential Privacy}
\label{algorithm}
\begin{algorithmic}[1]
\Require Sparse rating matrix $V \in \mathbb{R}^{n_m \times n_u}$ with observation mask $R \in \mathbb{R}^{n_m \times n_u}$; number of movies $n_m$; number of users $n_u$; privacy parameters $\epsilon_i \in (0,1)$, $\delta \in (0,1)$, target $\delta_r \in (0,1)$; range of ratings $\tau$; norm clipping parameter $C$; and step-size $\mu$.
\State Randomly initialize $X \in \mathbb{R}^{n_m \times n}, \Theta \in \mathbb{R}^{n_u \times n}$, where $\norm{x_i}_2 = \norm{\theta_j}_2 = 1$
\While {(not converged)}

    \State $\hat{V} \gets X\Theta^T \odot R$
    \State $\nabla_X \gets \left(\hat{V} - V\right)\Theta + \lambda X, \text{ where } \theta_j \gets \theta_j / \max\left(1, \frac{\norm{\theta_j}_2}{C}\right)$
    \State $\nabla_\Theta \gets \left(\hat{V} - V\right)^TX + \lambda \Theta, \text{ where } x_i \gets x_i / \max\left(1, \frac{\norm{x_i}_2}{C}\right)$
    \State Sample $\eta \in \mathbb{R}^{n_u \times n} \sim$ i.i.d $\mathcal{N}(0,\sigma^2)$, where $\sigma = \frac{\tau C}{\epsilon_i}\sqrt{2\log\frac{1.25}{\delta}}$
    \State $\hat{\nabla}_\Theta \gets \nabla_\Theta + \eta$
    \State $X \gets X - \mu \times \nabla_X$
    \State $\Theta \gets \Theta - \mu \times \hat{\nabla}_\Theta$

\EndWhile\\
\Return Differentially private user profile matrix $\Theta$
\end{algorithmic}
\end{algorithm}
\begin{theorem}[Privacy of Algorithm \ref{algorithm}]
\label{algorithmprivacy}
Consider Algorithm \ref{algorithm} in setting of the matrix factorization problem described in Section \ref{sec:MF} with $n_u$ user ratings and $n_m$ movies. Let the matrix $\Theta \in \mathbb{R}^{n_u \times n}$ encompass the user profile vectors and the matrix $X \in \mathbb{R}^{n_m \times n}$ encompass the movie (or item) profile vectors. Suppose that the range of ratings is $\tau$ and norm clipping parameter is $C$. Then Algorithm \ref{algorithm} guarantees $(\epsilon_\mathrm{opt}, \delta_r)$-DP for the released user profile matrix $\Theta$ for any $\delta_r \in (0,1)$, where $\epsilon_\mathrm{opt}$ is given by:
\begin{equation*}
\epsilon_\mathrm{opt} = \frac{J{\epsilon_i}^2}{4\log\frac{1.25}{\delta}} + 2\sqrt{\frac{J{\epsilon_i}^2 \log\frac{1}{\delta_r}}{4 \log\frac{1.25}{\delta}}}.
\end{equation*}
\end{theorem}
\begin{proof}
The proof of Theorem \ref{algorithmprivacy} follows from the application of the Gaussian mechanism \cite{dwork2014algorithmic} relating to Definition \ref{definition3} and Theorem \ref{theoremGuassian}, the post processing property \cite{dwork2014algorithmic} and the bound on  the $\mathcal{L}_2$ sensitivity of $\nabla_\Theta$. 
\par
Assuming a rating matrix $V^\prime$ that differs from \textit{V} by a single rating, we analytically find the expression for $\mathcal{L}_2$ sensitivity as: 
\begin{equation}
\label{sensitivity}
\begin{aligned}
\Delta_2\left(\nabla_\Theta\right)
&=\underset{V,V^\prime}{\max}\norm{\nabla_\Theta - {\nabla_\Theta}^\prime}_2\\
&=\underset{V,V^\prime}{\max} \lVert \left(X\Theta^T - V\right)^TX + \lambda \Theta \\& \hspace{12pt} - \left(X\Theta^T - V^\prime\right)^TX-\lambda \Theta \rVert_2\\
&=\underset{V,V^\prime}{\max}\norm{\Theta X^TX - V^TX - \Theta X^TX + {V^\prime}^TX}_2\\
&=\underset{V,V^\prime}{\max}\norm{\left(V^\prime - V\right)^TX}_2\\
&=\underset{V,V^\prime}{\max}\norm{
\begin{bmatrix}
0 &  \cdots & 0 \\
\vdots &  \ddots & \vdots \\
\left(v^\prime-v\right)x_{1,i} &  \ddots & \left(v^\prime-v\right)x_{n,i} \\
\vdots &  \ddots & \vdots \\
0 &  \cdots & 0 \\
\end{bmatrix}
}_2\\
&=\underset{V,V^\prime}{\max}\norm{\left(v^\prime-v\right)
\begin{bmatrix}
\horzbar & 0 & \horzbar \\
 & \vdots & \\
\horzbar & x_i & \horzbar \\
 & \vdots & \\
\horzbar & 0 & \horzbar \\
\end{bmatrix}
}_2\\
&= \underset{V,V^\prime}{\max} \hspace{2pt} |v^\prime - v|\norm{
\begin{bmatrix}
\horzbar & 0 & \horzbar \\
 & \vdots & \\
\horzbar & x_i & \horzbar \\
 & \vdots & \\
\horzbar & 0 & \horzbar \\
\end{bmatrix}
}_2 \\
\Delta_2\left(\nabla_\Theta\right) 
&= \underset{V,V^\prime}{\max} \hspace{2pt} |v^\prime - v|\norm{
x_i
}_2.
\end{aligned}
\end{equation}
\newline
We use Property 531 from \cite{petersen2008matrix} in line 7. Considering $\norm{x_i}_2 \le C$, and as $|v^\prime - v| \le \tau$, where $\tau$ is the range of ratings, we have $\Delta_2\left(\nabla_\Theta\right) = \tau C$, and the noise standard deviation in Line 6 of Algorithm \ref{algorithm} is $\sigma = \frac{\tau C}{\epsilon_i}\sqrt{2\log\frac{1.25}{\delta}}$. Recalling Definition \ref{definition3}, we can assert that the computation of $\nabla_\Theta$ at each step is at least $(\epsilon_i, \delta)$ differentially private. Therefore, the updated $\Theta$ at the end of the iterations is also differentially private. In the following, we show how to calculate the overall privacy spent using the propositions of RDP.

\noindent\textbf{Computing the Overall Privacy Risk.} Recall from Proposition ~\ref{prop3} stated in Section \ref{subsec:RDP} that a Gaussian mechanism, where the function has $\mathcal{L}_2$ sensitivity $\Delta_2\left(\nabla_\Theta\right)$ and which adds noise from $\mathcal{N}(0, \sigma^2)$ satisfies $(\alpha, \frac{\alpha {\Delta_2\left(\nabla_\Theta\right)}^2}{2\sigma^2})$-RDP. If an algorithm requires $J$ steps to converge, each of which satisfies $(\alpha, \frac{\alpha {\Delta_2\left(\nabla_\Theta\right)}^2}{2\sigma^2})$-RDP, then from Proposition~\ref{prop2}, the overall algorithm satisfies $(\alpha, \frac{\alpha J {\Delta_2\left(\nabla_\Theta\right)}^2}{2\sigma^2})$-RDP. We want to find the smallest value of $Overall \hspace{3pt} \epsilon$ (denoted as $\epsilon_\mathrm{opt}$ in Theorem~\ref{algorithmprivacy}), given a noise variance $\sigma^2$ and a target value of delta, $\delta_r$. Using Proposition~\ref{prop1} stated in Section~\ref{subsec:RDP}, we have that the overall algorithm is $(\epsilon, \delta_r)$-differentially private for any $\delta_r \in (0,1)$, where $\epsilon$ is given by
\begin{equation}
\label{overall_epsilon}
    \epsilon = \frac{\alpha J {\Delta_2\left(\nabla_\Theta\right)}^2}{2\sigma^2} + \frac{\log\frac{1}{\delta_r}}{\alpha-1}.
\end{equation}
To find the smallest $\epsilon$, we analytically to find $\alpha^*$ as,
\begin{align*}
\frac{\partial \epsilon}{\partial \alpha}
& = 0\\
\frac{J{\Delta_2\left(\nabla_\Theta\right)}^2}{2\sigma^2} - \frac{\log \frac{1}{\delta_r}}{(\alpha^* - 1)^2}
& = 0\\
\frac{J{\Delta_2\left(\nabla_\Theta\right)}^2}{2\sigma^2}
& = \frac{\log \frac{1}{\delta_r}}{(\alpha^* - 1)^2}\\
(\alpha^* - 1)^2
& = \frac{2\delta^2 \log\frac{1}{\delta_r}}{J{\Delta_2\left(\nabla_\Theta\right)}^2}\\
\alpha^*
& = 1 \pm \sqrt{\frac{2\sigma^2}{J{\Delta_2\left(\nabla_\Theta\right)}^2}\log\frac{1}{\delta_r}}.
\end{align*}
According to proposition \ref{prop1}, $\alpha > 1$. Therefore, we have
\begin{equation}
\label{alpha*}
\alpha^* = 1 + \sqrt{\frac{2\sigma^2}{J{\Delta_2\left(\nabla_\Theta\right)}^2}\log\frac{1}{\delta_r}}.
\end{equation}
Recall that for our setting, $\frac{\sigma}{\Delta_2\left(\nabla_\Theta\right)} = \frac{1}{\epsilon_i}\sqrt{2\log \frac{1.25}{\delta}}$. By substituting this in \eqref{alpha*}, and then replacing $\alpha$ in \eqref{overall_epsilon} by $\alpha^*$, we find the expression for the best (smallest) $\epsilon$ for a given target $\delta_r$ as:
\begin{equation}
\label{epsilon-opt}
\epsilon_\mathrm{opt} = \frac{J{\epsilon_i}^2}{4\log\frac{1.25}{\delta}} + 2\sqrt{\frac{J{\epsilon_i}^2 \log\frac{1}{\delta_r}}{4 \log\frac{1.25}{\delta}}}, \text{ for } \frac{J{\epsilon_i}^2 \log\frac{1}{\delta_r}}{4 \log\frac{1.25}{\delta}} > 0.
\end{equation}
\end{proof}

\section{Experimental Results and Observation}
\label{sec:exp}
We implement our algorithm on three datasets composed of real-world rating data. The first is the Movielens 1M \cite{movielens} dataset consisting of 1 million ratings of 3706 movies by 6040 users. This dataset has a rating density of 4.47\%. The second dataset is a slice taken from the Netflix Prize Data \cite{netflix} having 5.36 million ratings of 5466 movies by 11345 users with a rating density of 8.65\%. The third is another slice taken from the Anime Recommendations Database (later referred as AnimeReco) \cite{animereco} comprising 1.54 million ratings of 2772 anime by 4623 users with a rating density of 12.03\%. Note that the Netflix and AnimeReco dataset slices were conditioned to have users with ratings over 200 and under 2500 and movies with ratings over 100. The ratings in the Movielens and Netflix datasets are integers that range from 1 to 5, whereas the AnimeReco dataset has integer ratings ranging from 1 to 10.
\begin{figure*}[t]
  \centering
  \begin{minipage}[t]{0.42\textwidth}
    \includegraphics[width=\textwidth]{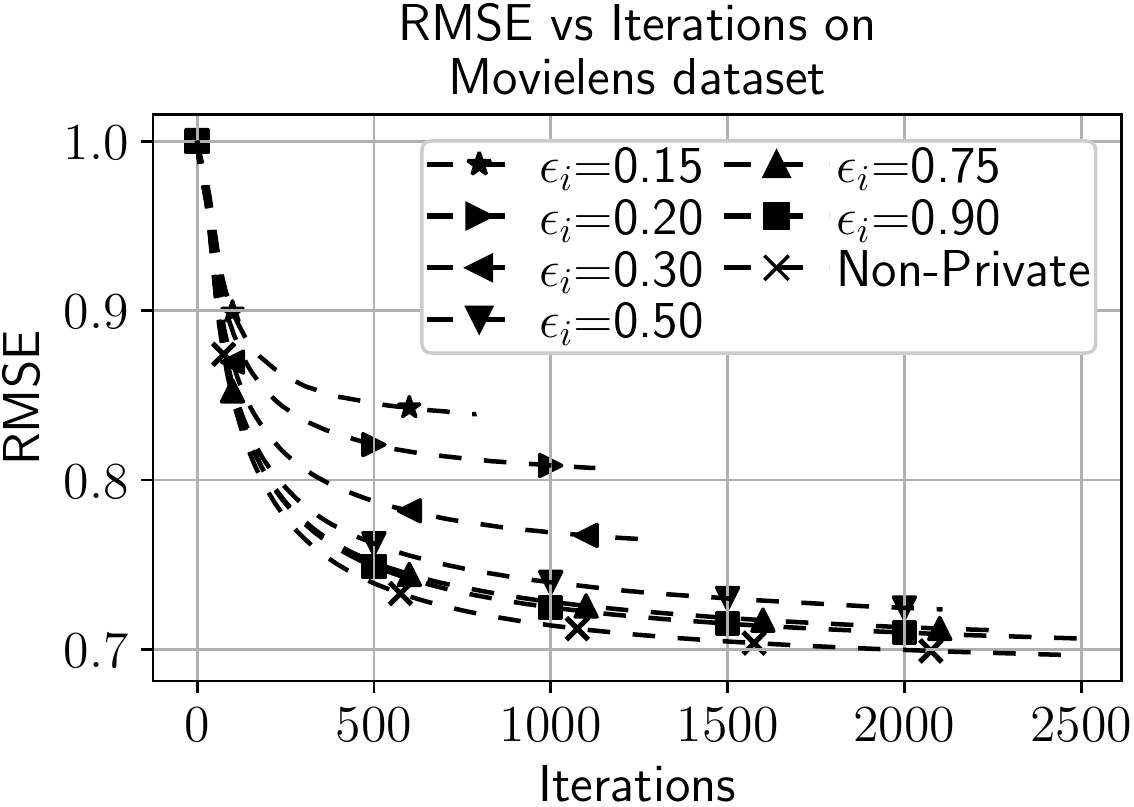}
    \caption{RMSE vs Iterations on Movielens dataset}
    \label{fig:1}
  \end{minipage}
  \hfill
  \begin{minipage}[t]{0.275\textwidth}
    \includegraphics[width=\textwidth]{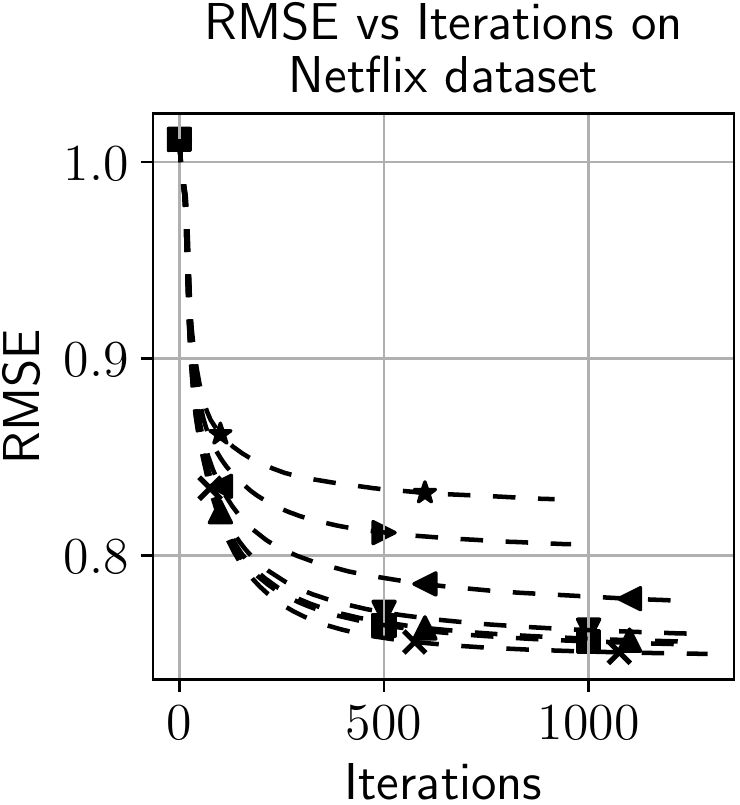}
    \caption{RMSE vs Iterations on Netflix dataset}
    \label{fig:2}
  \end{minipage}
  \hfill
  \begin{minipage}[t]{0.275\textwidth}
    \includegraphics[width=\textwidth]{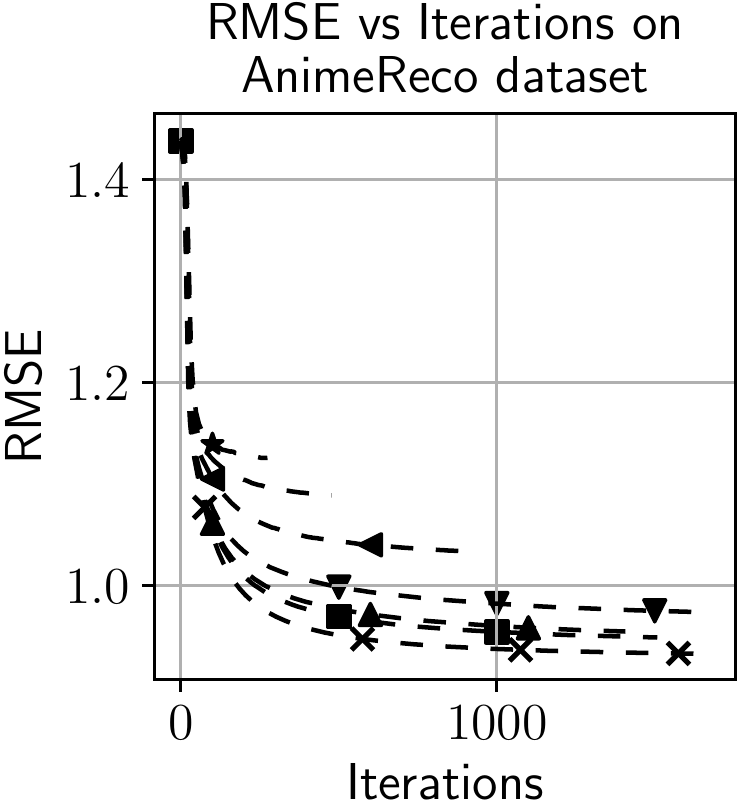}
    \caption{RMSE vs Iterations on AnimeReco dataset}
    \label{fig:3}
  \end{minipage}
\end{figure*}

\begin{figure*}[t]
  \centering
  \begin{minipage}[t]{0.325\textwidth}
    \includegraphics[width=\textwidth]{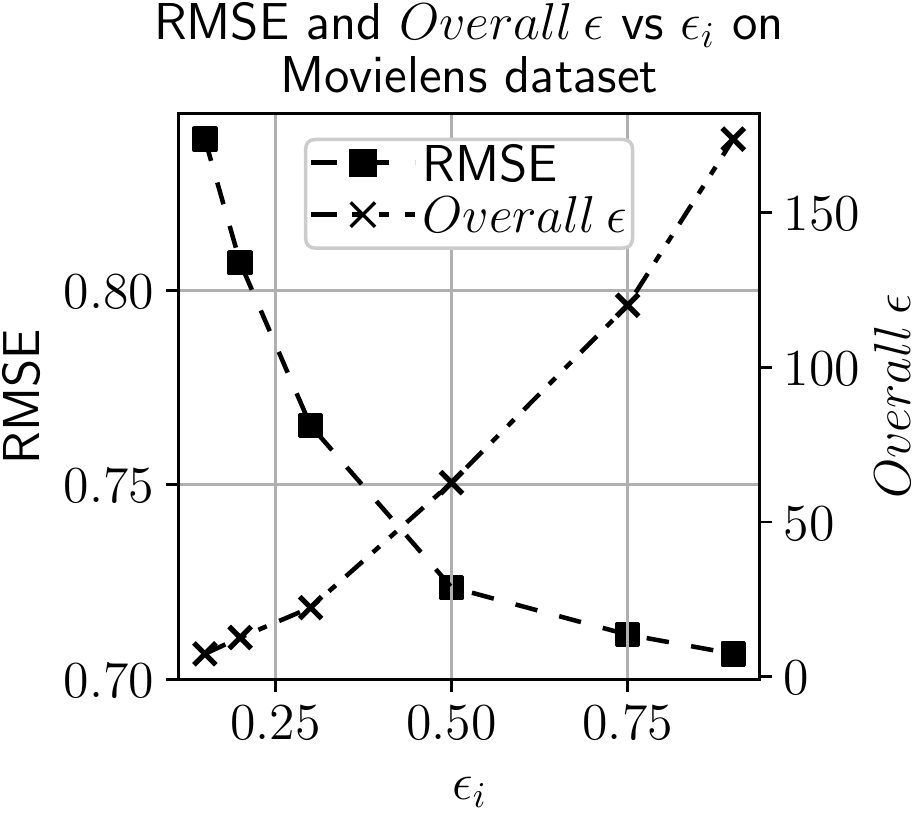}
    \caption{RMSE and $Overall \hspace{3pt} \protect \epsilon$ vs $\protect \epsilon_i$ on Movielens dataset}
    \label{fig:4}
  \end{minipage}
  \hfill
  \begin{minipage}[t]{0.325\textwidth}
    \includegraphics[width=\textwidth]{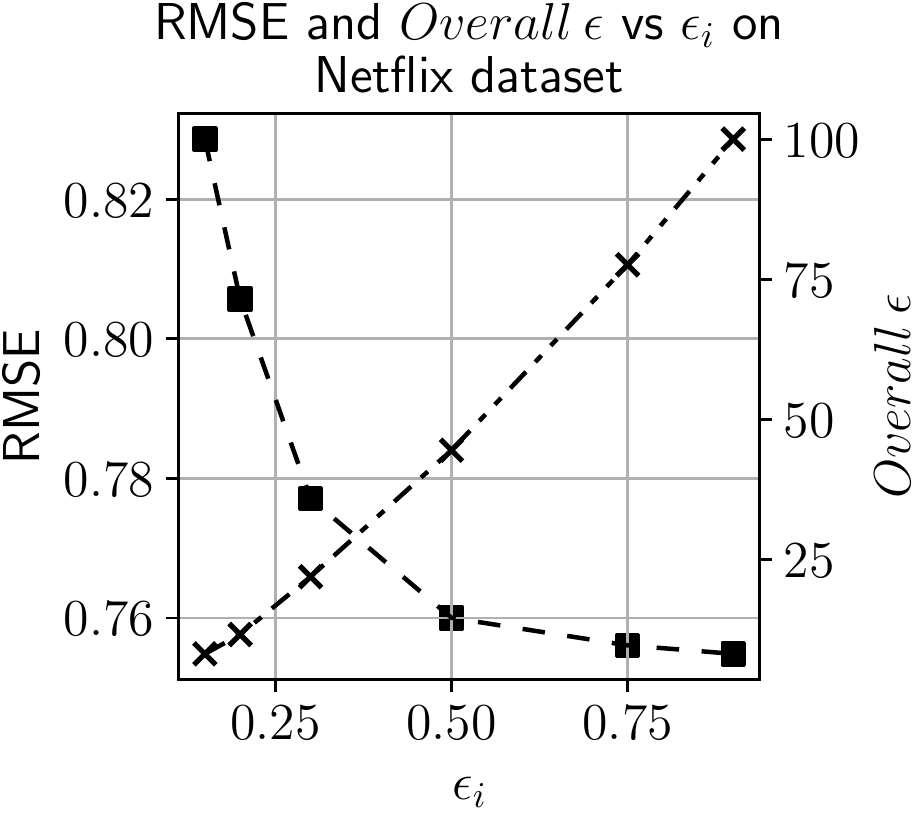}
    \caption{RMSE and $Overall \hspace{3pt} \protect \epsilon$ vs $\protect \epsilon_i$ on Netflix dataset}
    \label{fig:5}
  \end{minipage}
  \hfill
  \begin{minipage}[t]{0.325\textwidth}
    \includegraphics[width=\textwidth]{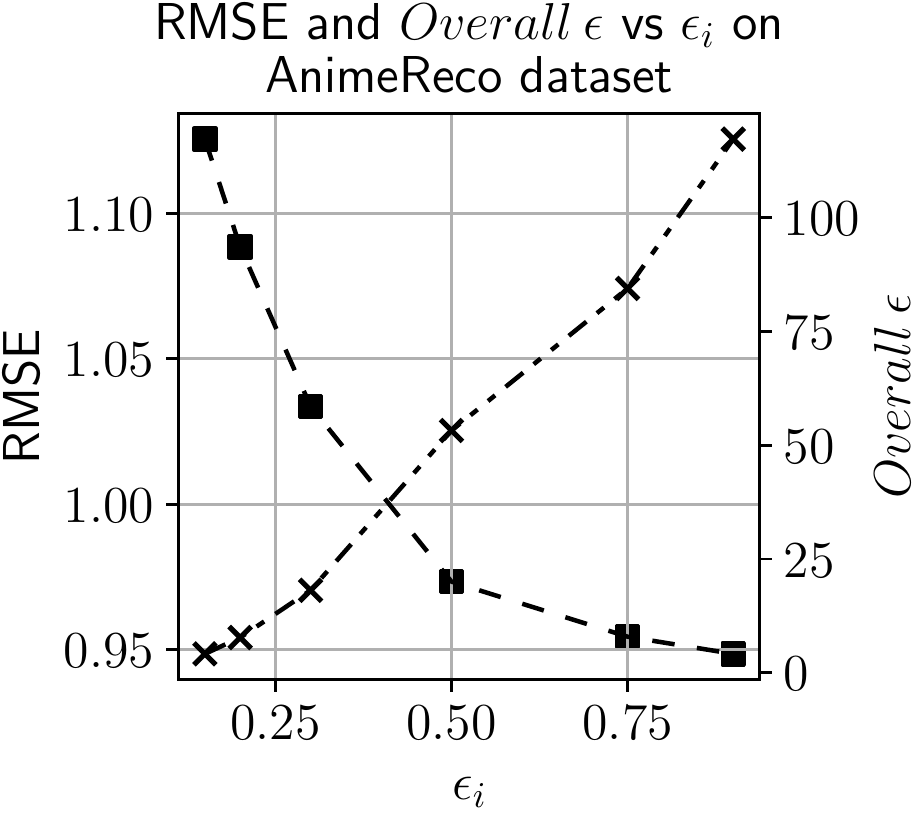}
    \caption{RMSE and $Overall \hspace{3pt} \protect \epsilon$ vs $\protect \epsilon_i$ on AnimeReco dataset}
    \label{fig:6}
  \end{minipage}
\end{figure*}
\par
All of the following experiments are conducted using Google's Colab Pro service which provides a Tesla P100 GPU with 16280 MiB of memory and 12.68 GB of RAM. For our experiments, we use $\sigma = \frac{\tau C}{\epsilon_i}\sqrt{2\log\frac{1.25}{\delta}}$. Here $\sigma$ is the parameter of the Gaussian mechanism, and $(\epsilon_i, \delta)$ are the privacy parameters per iteration. $C$ is the norm clipping parameter. We use $C = 1$ throughout our experiments. $Overall \hspace{3pt} \epsilon$, and $\delta_r$ are the overall privacy risk and target delta of the algorithm, respectively. Recall that $Overall \hspace{3pt} \epsilon$ is denoted as $\epsilon_\mathrm{opt}$ in Theorem \ref{algorithmprivacy}. \emph{Target} $\delta_r$ is the probability that the algorithm fails to satisfy the privacy definition. We use Root Mean Square Error (RMSE) as an indicator of utility for our algorithm. We investigate the impact of differential privacy on the utility of the recommender system. We observe the learning curve from the RMSE vs. Iteration plots of figures \ref{fig:1} -- \ref{fig:3}. For low values of $\epsilon_i$, for example, $\epsilon_i = 0.15$, more privacy is maintained in exchange for more RMSE. We keep all other parameters unchanged, and with increased values of $\epsilon_i$, the learning curves move closer to the non-private learning curve. We can achieve very close utility compared to the non-private utility. Step-size $\mu = 0.0005$ and profile vector dimension $n = 20$ are used for this experiment. Note that, for all the experiments discussed in this paper, $\delta = 0.01$, and $\delta_r = 0.00001$ are used. The RMSE curve approximates the non-private RMSE curve even better for $\epsilon_i \geq 0.5$ values. However, at this stage, the decrease in error is not worth compared to the substantially increasing overall privacy risk. We evidently observe behavior suggesting trade-off between utility, and privacy risk from figures \ref{fig:4} -- \ref{fig:6}. We notice an optimal trade-off between the utility and overall privacy risk (at around $\epsilon_i = 0.4$). This implicitly means that we can implement an accurate enough recommender system, which satisfies strict privacy guarantees while providing utility very close to that of a non-private recommender. 
\begin{figure*}[t]
  \centering
  \begin{minipage}[t]{0.42\textwidth}
    \includegraphics[width=\textwidth]{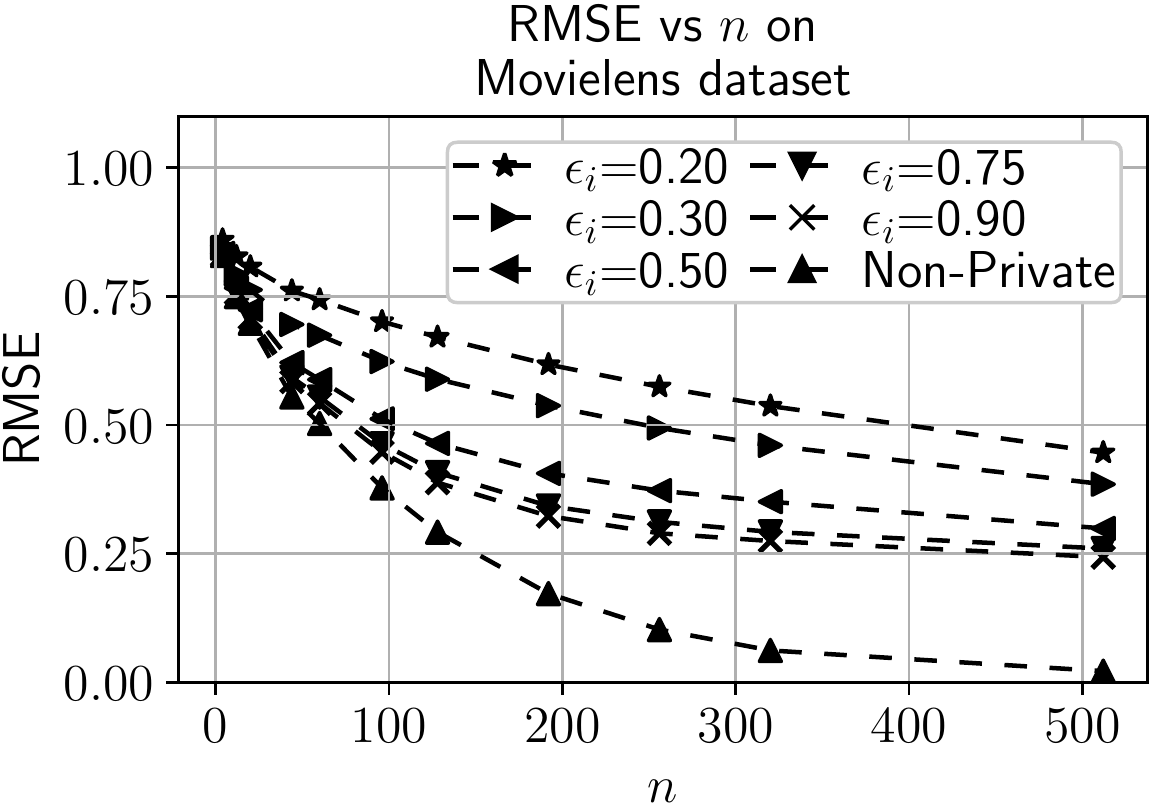}
    \caption{RMSE vs $\protect n$ on Movielens dataset}
    \label{fig:7}
  \end{minipage}
  \hfill
  \begin{minipage}[t]{0.27\textwidth}
    \includegraphics[width=\textwidth]{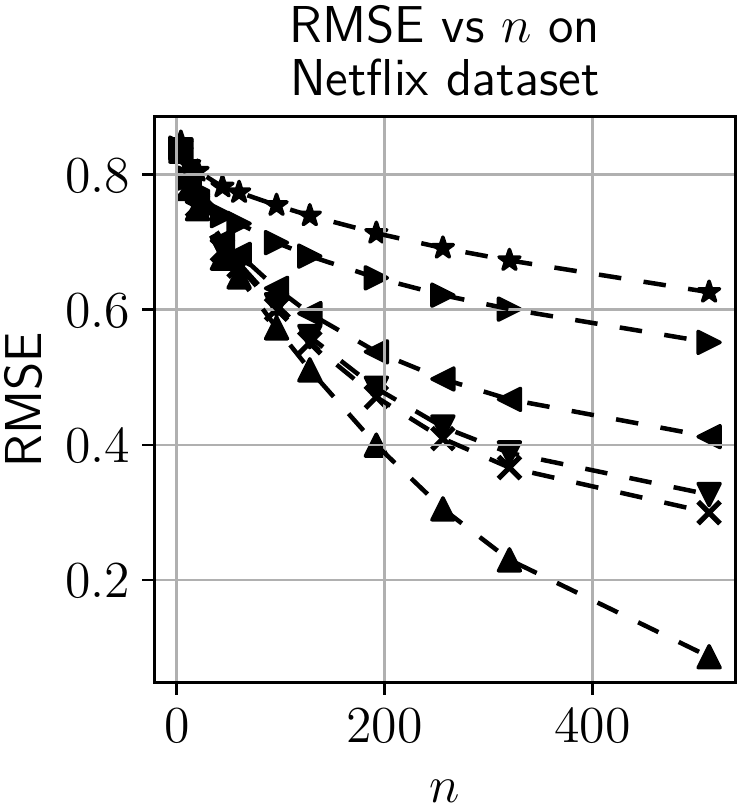}
    \caption{RMSE vs $\protect n$ on Netflix dataset}
    \label{fig:8}
  \end{minipage}
  \hfill
  \begin{minipage}[t]{0.275\textwidth}
    \includegraphics[width=\textwidth]{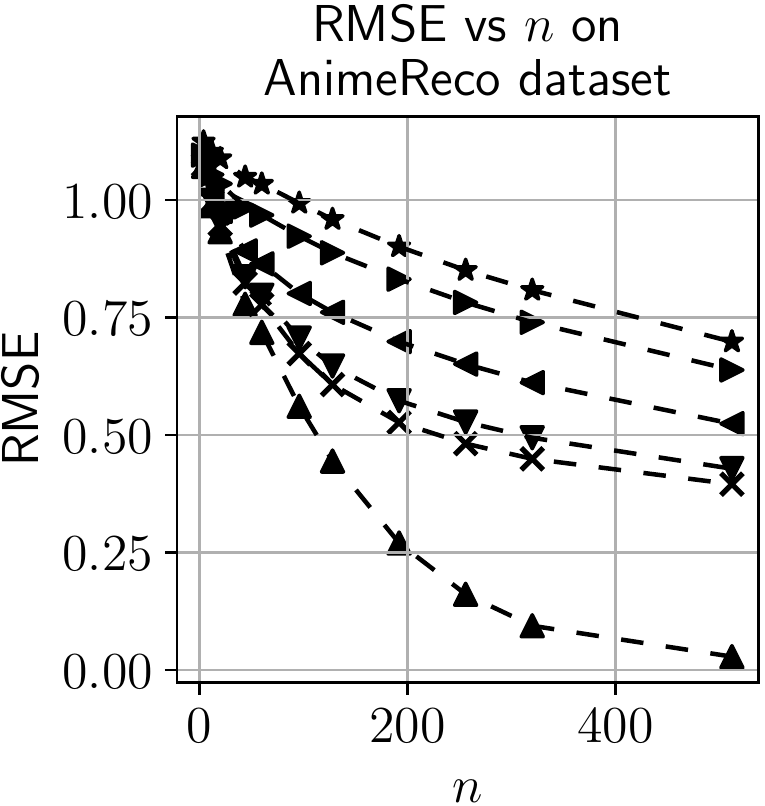}
    \caption{RMSE vs $\protect n$ on AnimeReco dataset}
    \label{fig:9}
  \end{minipage}
\end{figure*}

\begin{figure*}[t]
  \centering
  \begin{minipage}[t]{0.42\textwidth}
    \includegraphics[width=\textwidth]{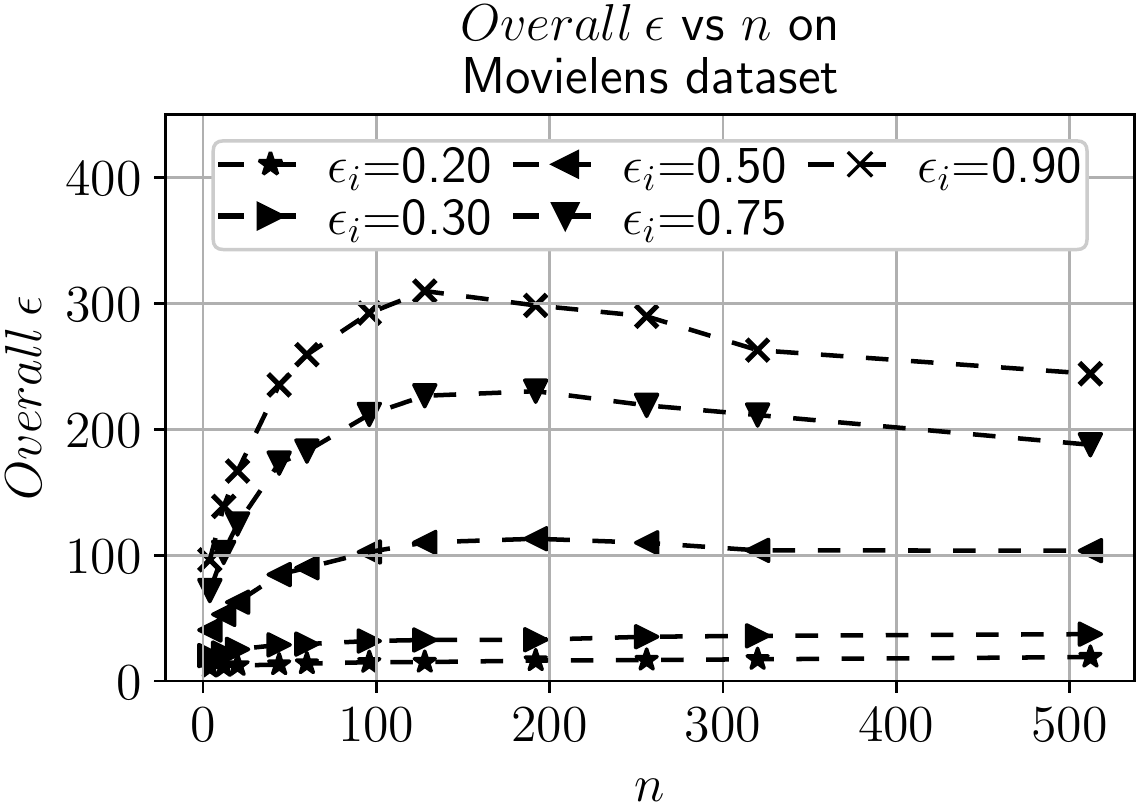}
    \caption{$Overall \hspace{3pt} \protect \epsilon$ vs $\protect n$ on Movielens dataset}
    \label{fig:10}
  \end{minipage}
  \hfill
  \begin{minipage}[t]{0.28\textwidth}
    \includegraphics[width=\textwidth]{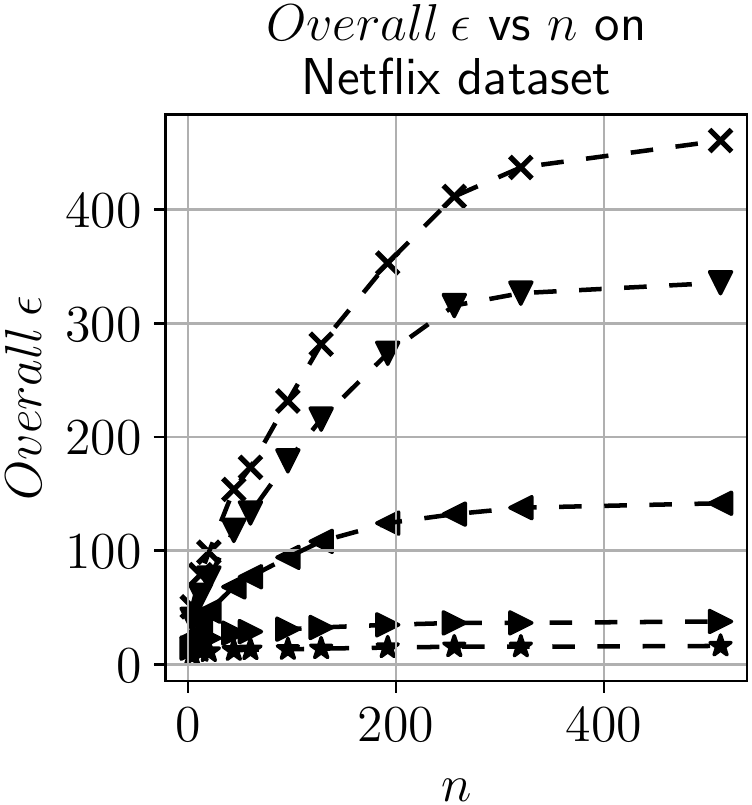}
    \caption{$Overall \hspace{3pt} \protect \epsilon$ vs $\protect n$ on Netflix dataset}
    \label{fig:11}
  \end{minipage}
  \hfill
  \begin{minipage}[t]{0.28\textwidth}
    \includegraphics[width=\textwidth]{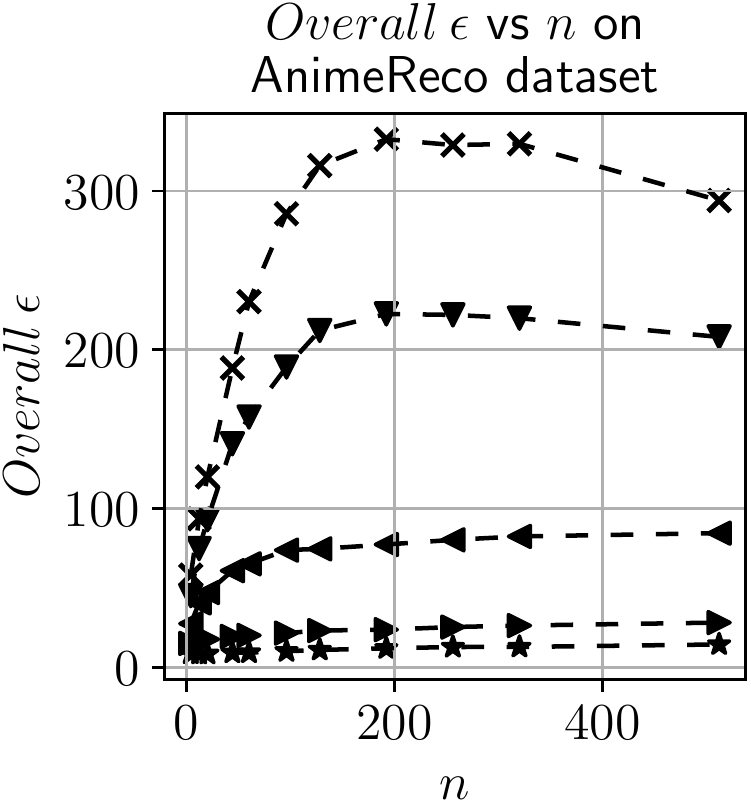}
    \caption{$Overall \hspace{3pt} \protect \epsilon$ vs $\protect n$ on AnimeReco dataset}
    \label{fig:12}
  \end{minipage}
\end{figure*}

\begin{figure*}[t]
  \centering
  \begin{minipage}[t]{0.41\textwidth}
    \includegraphics[width=\textwidth]{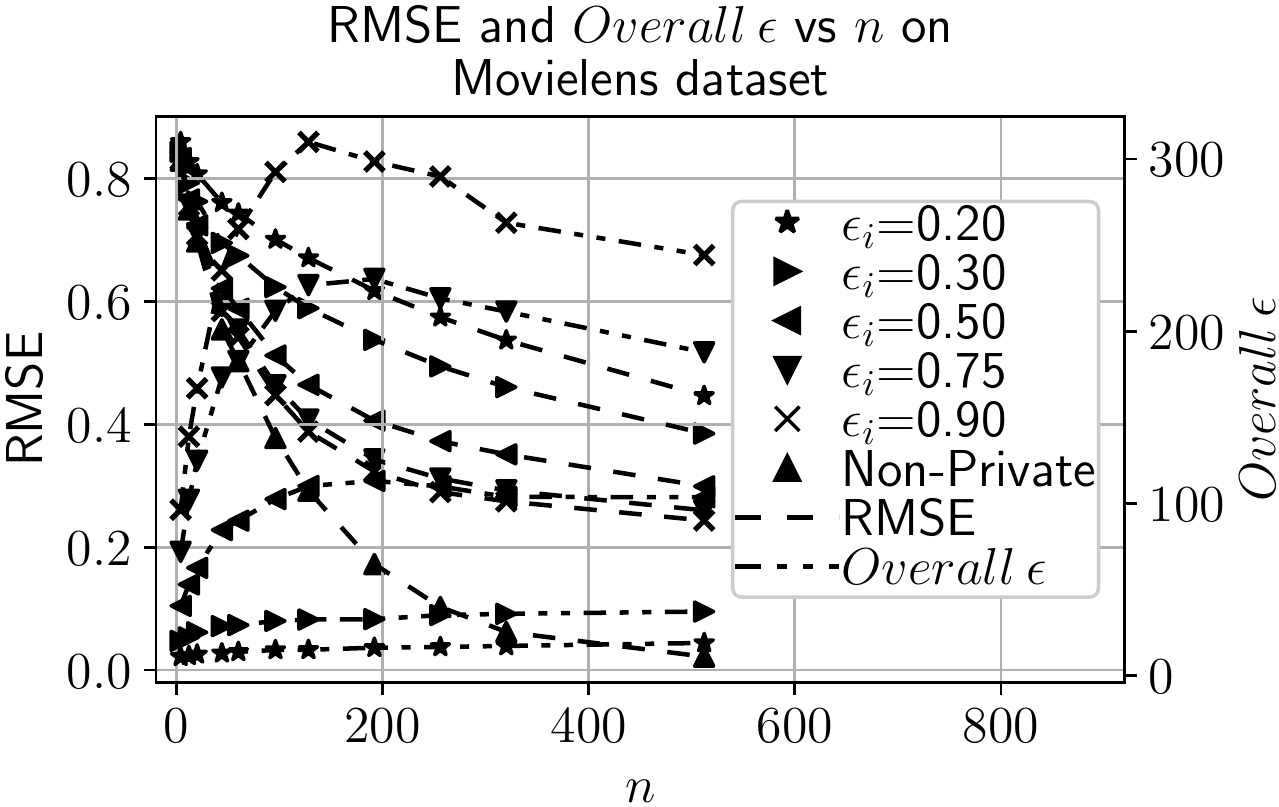}
    \caption{RMSE and $Overall \hspace{3pt} \protect \epsilon$ vs $\protect n$ on Movielens dataset}
    \label{fig:13}
  \end{minipage}
  \hfill
  \begin{minipage}[t]{0.283\textwidth}
    \includegraphics[width=\textwidth]{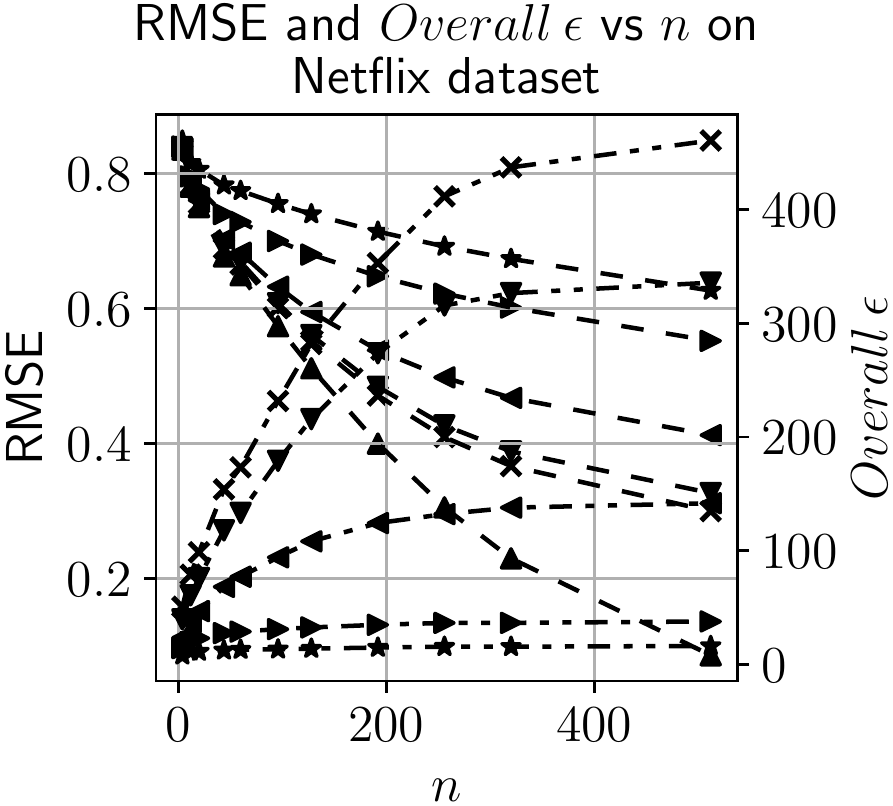}
    \caption{RMSE and $Overall \hspace{3pt} \protect \epsilon$ vs $\protect n$ on Netflix dataset}
    \label{fig:14}
  \end{minipage}
  \hfill
  \begin{minipage}[t]{0.287\textwidth}
    \includegraphics[width=\textwidth]{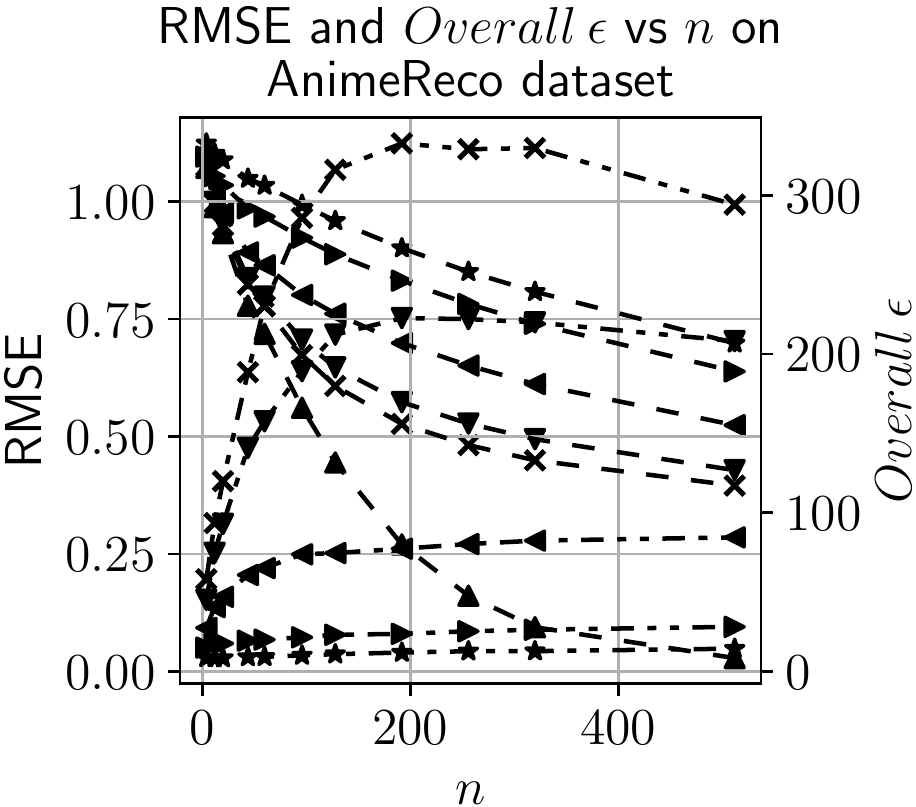}
    \caption{RMSE and $Overall \hspace{3pt} \protect \epsilon$ vs $\protect n$ on AnimeReco dataset}
    \label{fig:15}
  \end{minipage}
\end{figure*}

\begin{figure*}[t]
  \centering
  \begin{minipage}[t]{0.325\textwidth}
    \includegraphics[width=\textwidth]{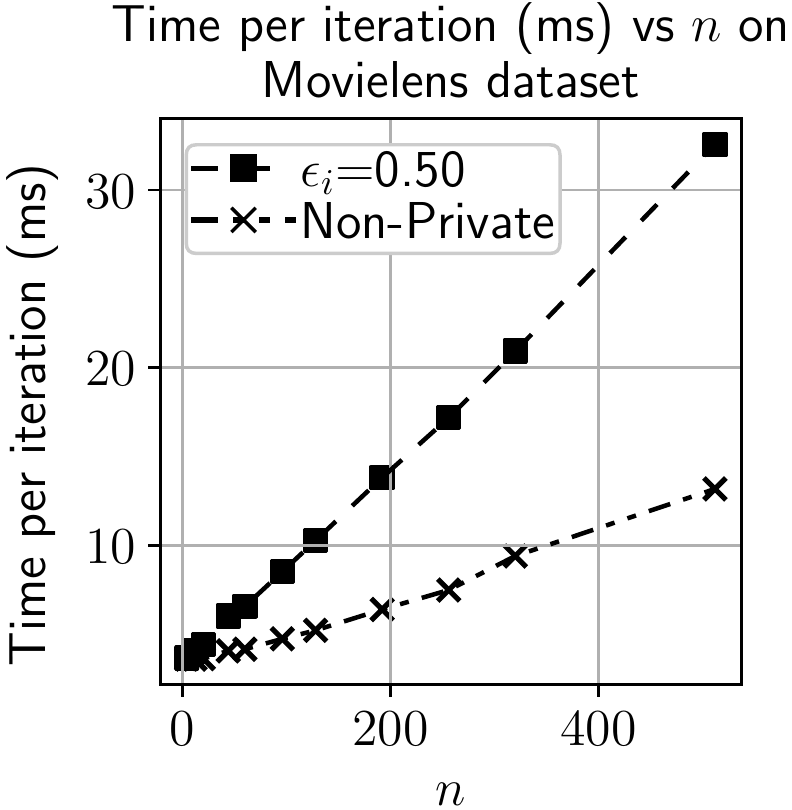}
    \caption{Time per iteration (ms) vs $\protect n$ on Movielens dataset}
    \label{fig:16}
  \end{minipage}
  \hfill
  \begin{minipage}[t]{0.325\textwidth}
    \includegraphics[width=\textwidth]{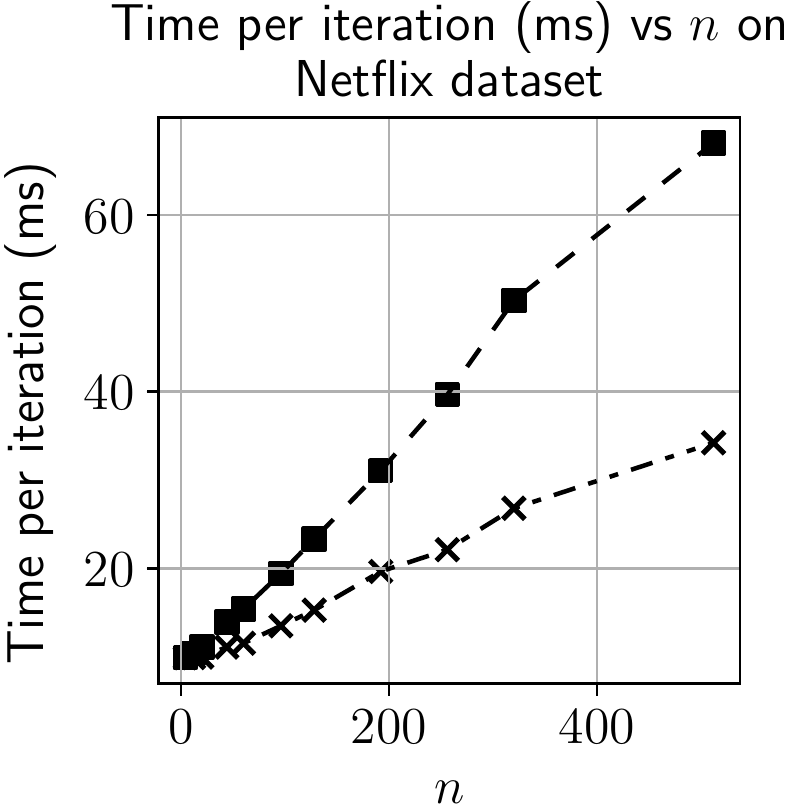}
    \caption{Time per iteration (ms) vs $\protect n$ on Netflix dataset}
    \label{fig:17}
  \end{minipage}
  \hfill
  \begin{minipage}[t]{0.325\textwidth}
    \includegraphics[width=\textwidth]{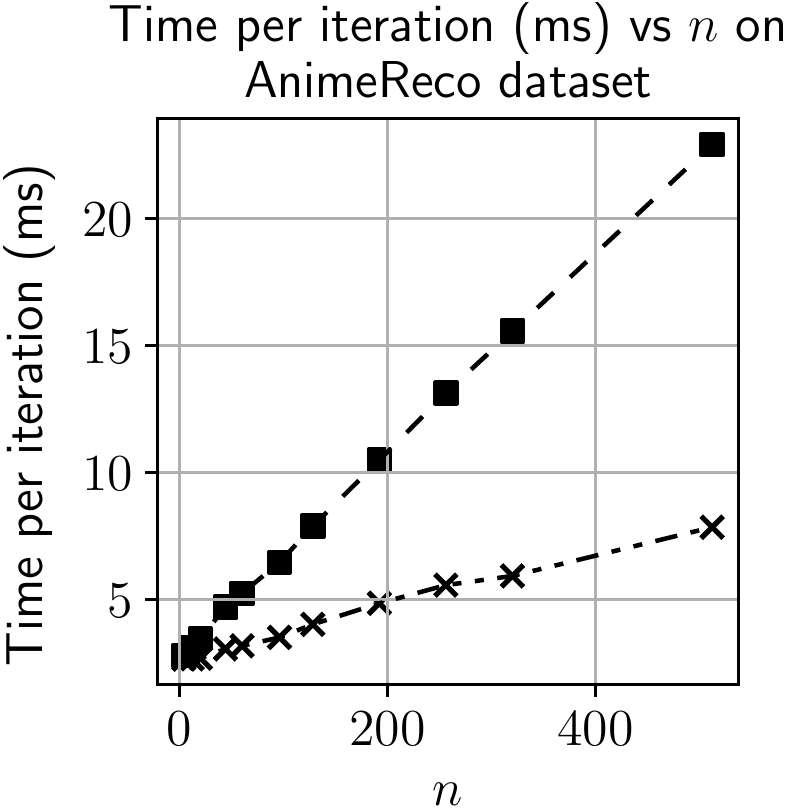}
    \caption{Time per iteration (ms) vs $\protect n$ on AnimeReco dataset}
    \label{fig:18}
  \end{minipage}
\end{figure*}
\par
In the following experiment, we observe the effect of the dimensionality $n$ of profile vectors on private and non-private training. The step-size $\mu$ was kept constant at 0.0005. Here, we notice from figures \ref{fig:7} -- \ref{fig:9} that the non-private recommender system reaches very close to zero RMSE for high values of $n$, and the non-private RMSE diverges quite noticeably from differentially private (DP) RMSE curves. The DP-RMSE curves show an indication of saturation at higher errors. A possible reason for such behavior can be the fact that non-private MF permits the algorithm to learn profile vectors that describe the users and their ratings the best. Consequently, a profile vector of a higher dimension can learn more about a user, enabling it to reach almost zero error conditions. On the other hand, differentially private MF puts a restraint on the learning process that prevents it from being too sensitive to the user's ratings. Differential privacy only allows the algorithm to extract aggregate information. Now, if we look at the $Overall \hspace{3pt} \epsilon$ vs. $n$ plots in figures \ref{fig:10} -- \ref{fig:12}, for $\epsilon_i = 0.5$ and below, the overall privacy risk tends to become invariant with respect to the increasing values of $n$. For higher $\epsilon_i$ values, $Overall \hspace{3pt} \epsilon$ tends to reach a high peak and then slightly decline with increasing profile vector dimensionality. Using $\epsilon_i$ values higher than 0.5 is evidently not useful as it increases privacy risk too much compared to the benefit of lower RMSE it provides. Moreover, it is not justified to use very high values on $n$ just because of the advantage of lower RMSE in exchange for an almost negligible increase in privacy risk. The reason is that for values of $n$ such as 320 or higher, the computation time per iteration can be several times longer than the time needed for values of $n$ such as 20 or lower as depicted in figures \ref{fig:16} -- \ref{fig:18}. The times per iteration in the figures are taken as the mean time for 300 iterations for each configuration.
\begin{figure*}[t]
  \centering
  \begin{minipage}[t]{0.40\textwidth}
    \includegraphics[width=\textwidth]{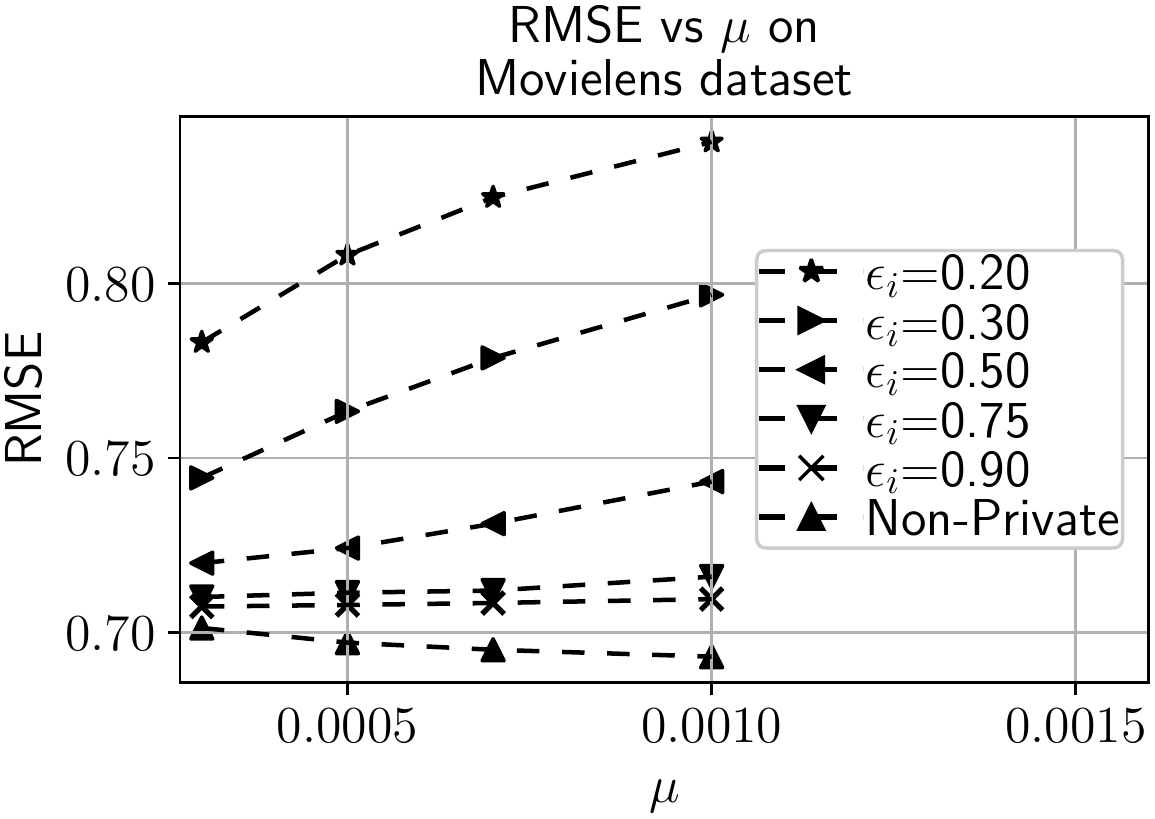}
    \caption{RMSE vs $\protect \mu$ on Movielens dataset}
    \label{fig:19}
  \end{minipage}
  \hfill
  \begin{minipage}[t]{0.29\textwidth}
    \includegraphics[width=\textwidth]{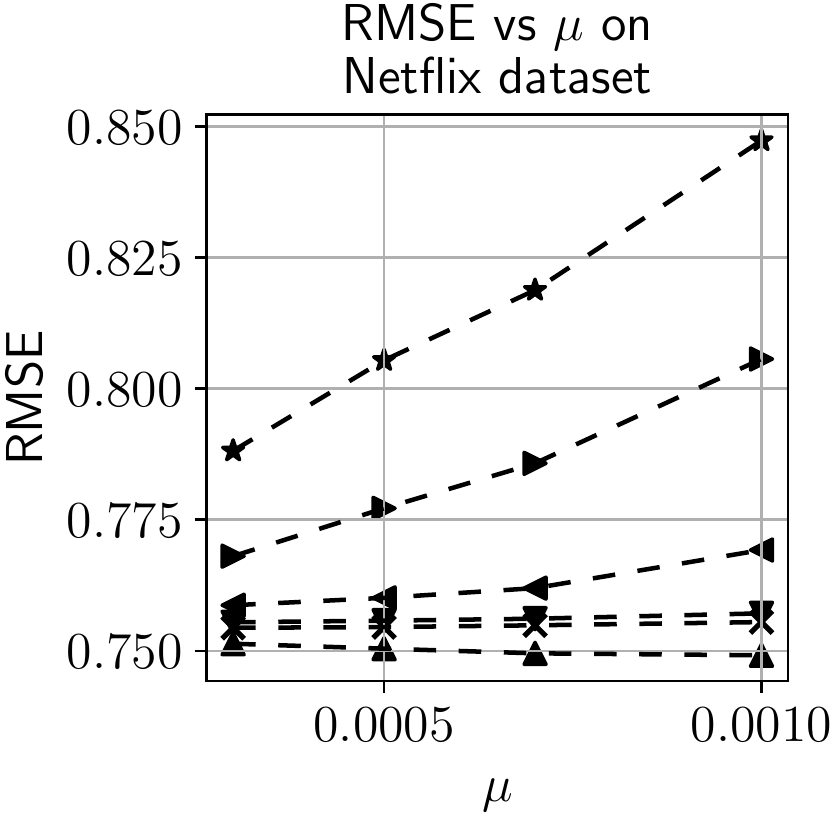}
    \caption{RMSE vs $\protect \mu$ on Netflix dataset}
    \label{fig:20}
  \end{minipage}
  \hfill
  \begin{minipage}[t]{0.28\textwidth}
    \includegraphics[width=\textwidth]{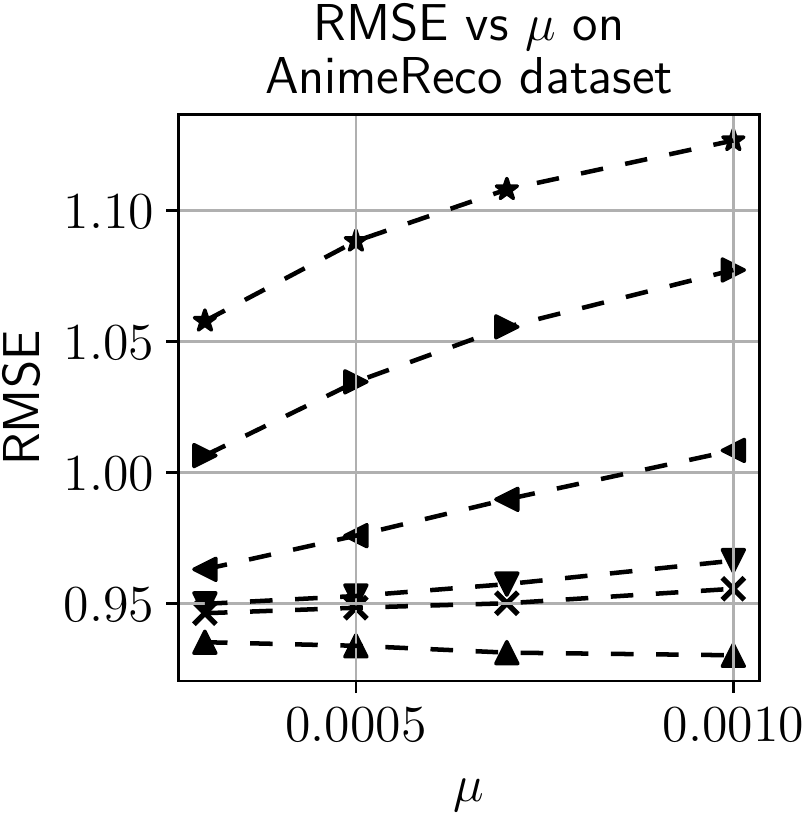}
    \caption{RMSE vs $\protect \mu$ on AnimeReco dataset}
    \label{fig:21}
  \end{minipage}
\end{figure*}

\begin{figure*}[t]
  \centering
  \begin{minipage}[t]{0.41\textwidth}
    \includegraphics[width=\textwidth]{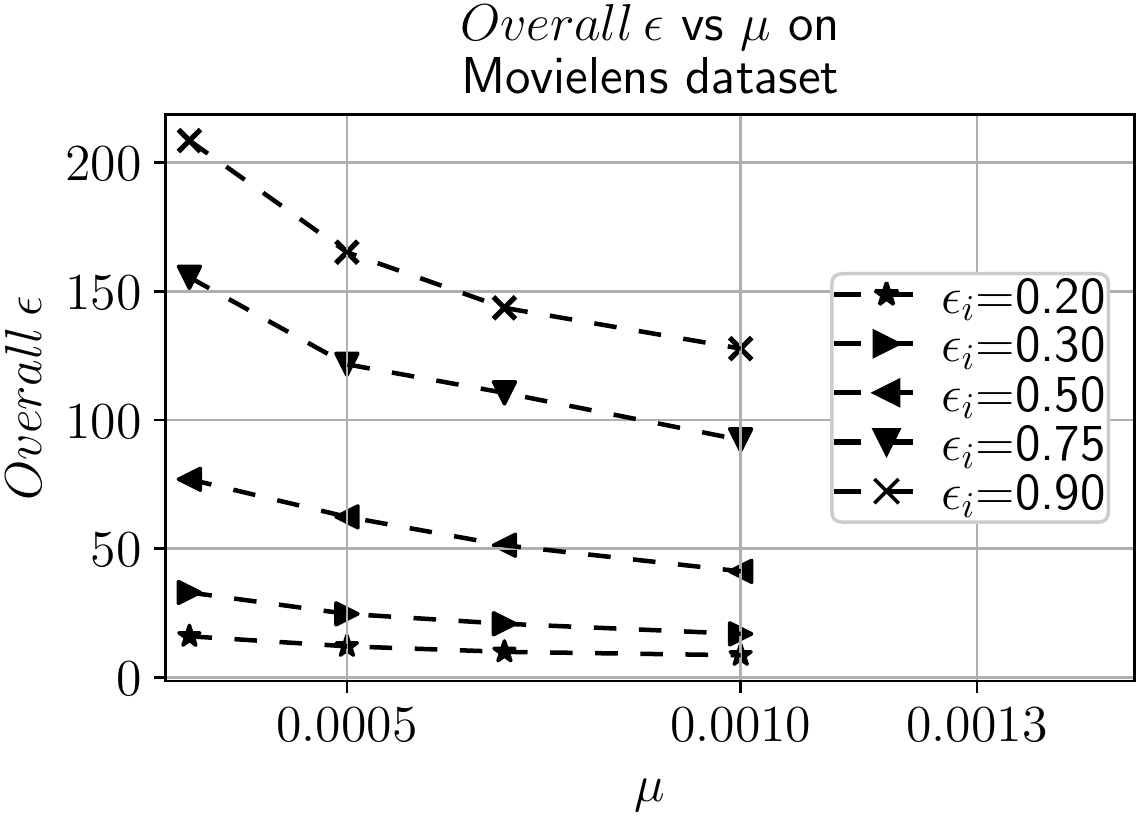}
    \caption{$Overall \hspace{3pt} \protect \epsilon$ vs $\protect \mu$ on Movielens dataset}
    \label{fig:22}
  \end{minipage}
  \hfill
  \begin{minipage}[t]{0.285\textwidth}
    \includegraphics[width=\textwidth]{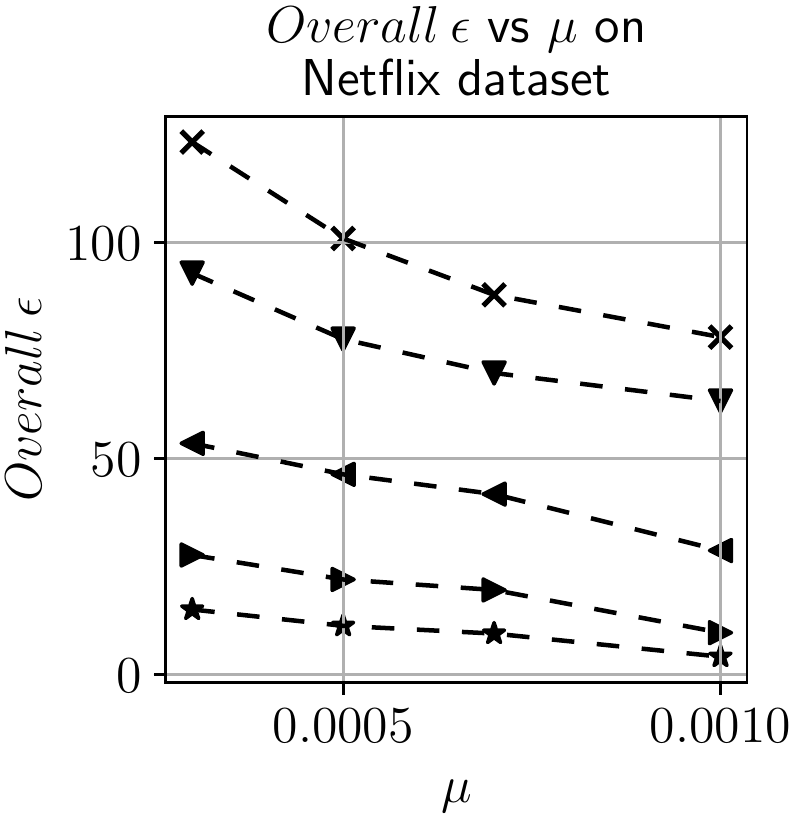}
    \caption{$Overall \hspace{3pt} \protect \epsilon$ vs $\protect \mu$ on Netflix dataset}
    \label{fig:23}
  \end{minipage}
  \hfill
  \begin{minipage}[t]{0.285\textwidth}
    \includegraphics[width=\textwidth]{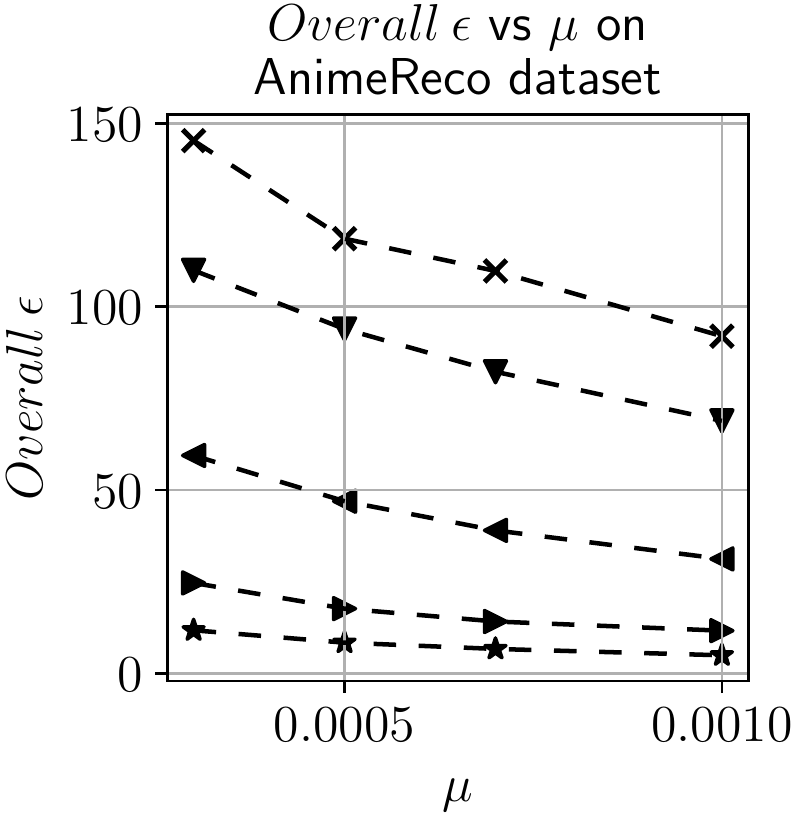}
    \caption{$Overall \hspace{3pt} \protect \epsilon$ vs $\protect \mu$ on AnimeReco dataset}
    \label{fig:24}
  \end{minipage}
\end{figure*}

\begin{figure*}[t]
  \centering
  \begin{minipage}[t]{0.41\textwidth}
    \includegraphics[width=\textwidth]{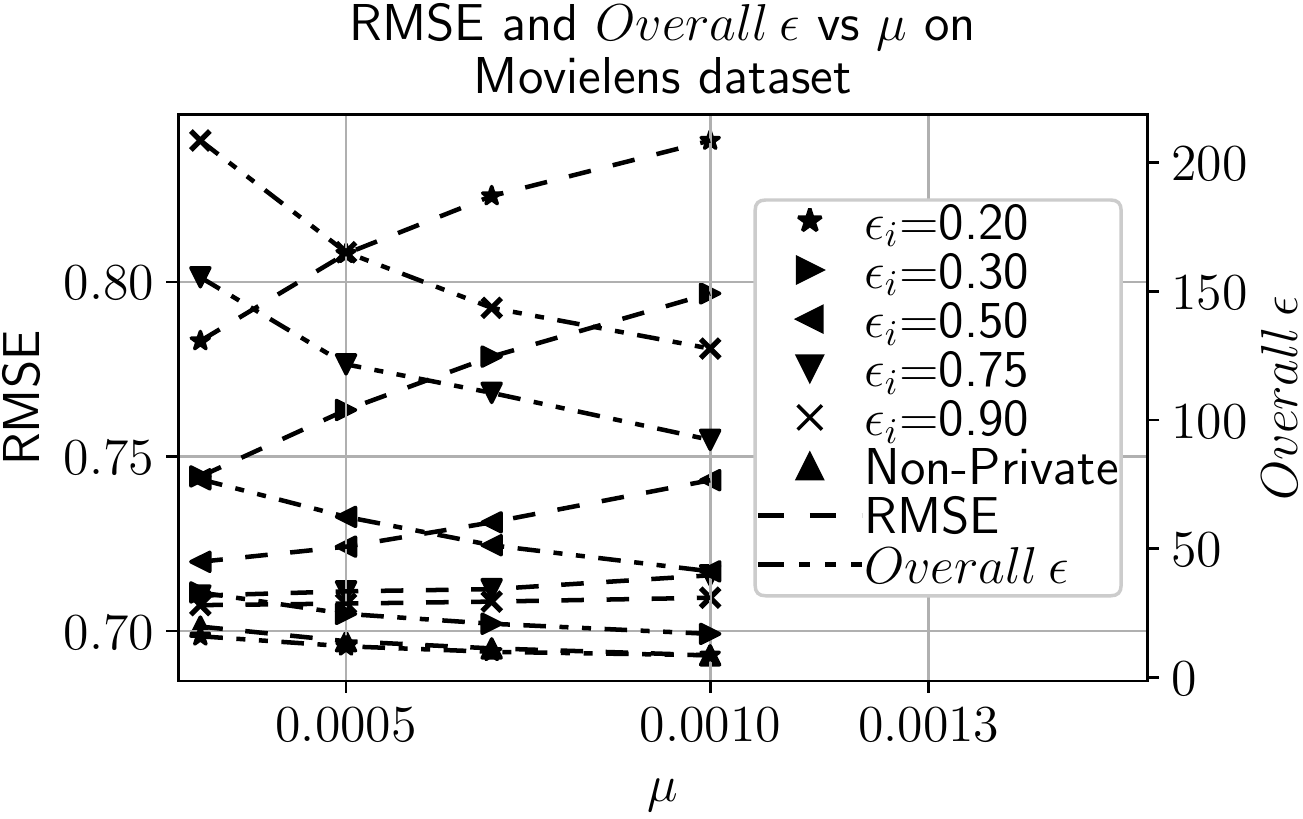}
    \caption{RMSE and $Overall \hspace{3pt} \protect \epsilon$ vs $\protect \mu$ on Movielens dataset}
    \label{fig:25}
  \end{minipage}
  \hfill
  \begin{minipage}[t]{0.29\textwidth}
    \includegraphics[width=\textwidth]{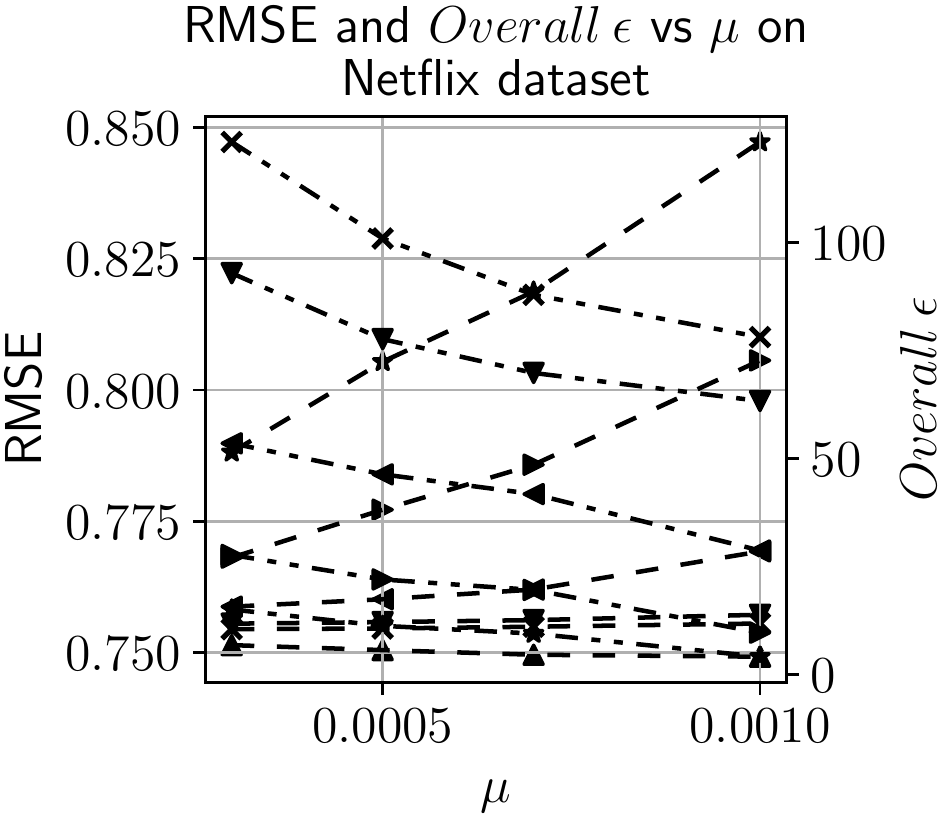}
    \caption{RMSE and $Overall \hspace{3pt} \protect \epsilon$ vs $\protect \mu$ on Netflix dataset}
    \label{fig:26}
  \end{minipage}
  \hfill
  \begin{minipage}[t]{0.28\textwidth}
    \includegraphics[width=\textwidth]{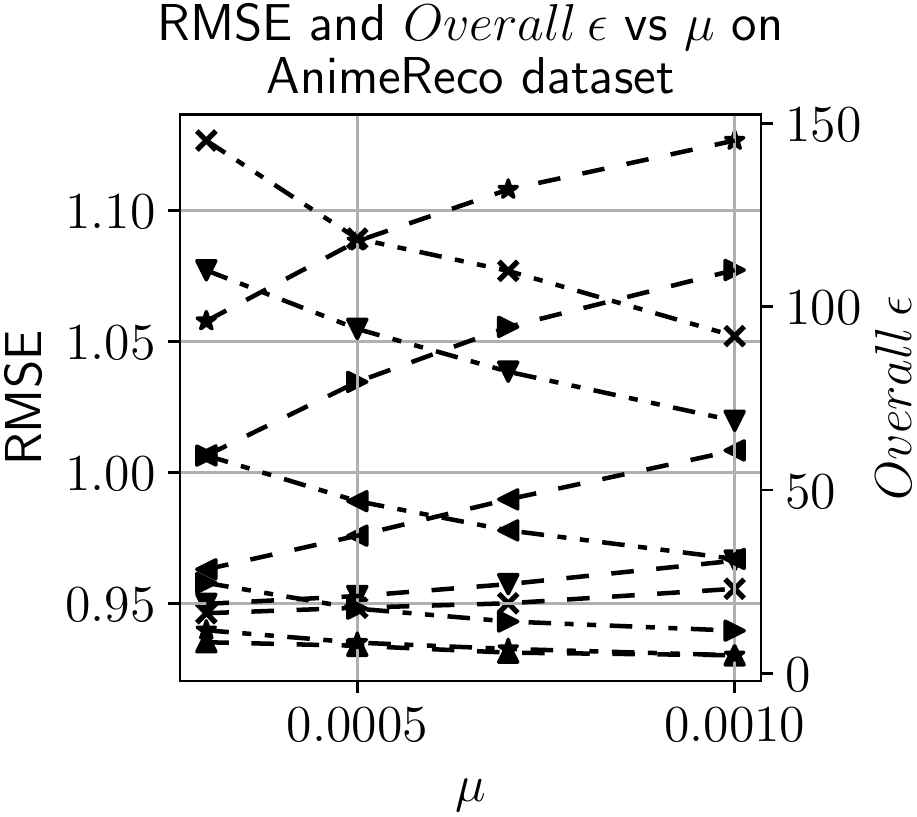}
    \caption{RMSE and $Overall \hspace{3pt} \protect \epsilon$ vs $\protect \mu$ on AnimeReco dataset}
    \label{fig:27}
  \end{minipage}
\end{figure*}
\par
Through the final experiment, we uncover the effects of increasing step-size $\mu$ on the learning process. The profile vector dimension $n$ was kept at 20 for this experiment. By observing figures \ref{fig:19} -- \ref{fig:24}, we find that, using larger step-size results in lower $Overall \hspace{3pt} \epsilon$ and higher prediction errors. From \eqref{overall_epsilon}, we know that $Overall \hspace{3pt} \epsilon$ is directly proportional to the number of steps the training took. As all the other parameters required to calculate $Overall \hspace{3pt} \epsilon$ were unchanged, decreasing values indicate that the training took fewer steps. The cause of this behavior is that too large a step-size can cause gradient descent to miss the minima of the loss function. It can even cause gradient descent to start diverging. As incorporating the Gaussian mechanism involves adding Gaussian noise to the target function, the gradient, in our case, the descent is not as smooth as before. Consequently, the process cannot handle a relatively large step-size. Larger step-sizes on top of noisy gradients may have caused the descent to miss the desired minima, resulting in higher prediction errors. All these characteristics we have investigated can be utilized to pick the necessary parameter values to develop differentially private recommendation systems as per accuracy and privacy requirements.

\section{Conclusion}
\label{sec:conclusion}
In this paper, we proposed a differentially private matrix factorization based recommendation system via the Gaussian mechanism. We showed that our algorithm provides an $(\epsilon_\mathrm{opt}, \delta_r)$ differentially private user profile matrix containing the inferred user profile vectors, which can then be used with appropriate movie (or item) profile vectors for the purpose of recommending movies (or items) to the user. We observed different characteristics of our proposed algorithm that can be utilized for formulating a practical privacy-preserving recommendation system. We demonstrated the superiority of our algorithm with varying privacy levels and other key parameters on three real datasets. We also showed that our algorithm's utility can closely match that of the non-private algorithm for certain parameter choices. Recognizing the inherent issue of the \emph{spent privacy budget} of multi-shot algorithms, we provided a tighter characterization of the overall privacy budget of our algorithm using the R\'enyi Differential Privacy. An interesting future work could be to develop computationally efficient matrix factorization that works well with differential privacy. Another promising future work could be to introduce personalized differential privacy, proposed in \cite{liu2015fast}, to our work; and adapting our proposed method to improved versions of MF, such as the Weighted Non-negative Matrix Factorization \cite{hua2015differentially}.

% \appendix

% \newpage
\bibnote{movielens}{\emph{[Dataset]}}
\bibnote{netflix}{\emph{[Dataset]}}
\bibnote{animereco}{\emph{[Dataset]}}

\bibliographystyle{elsarticle-num} 
\bibliography{cas-refs}





% \end{thebibliography}
\newpage
\begin{wrapfigure}{l}{25mm} 
    \includegraphics[width=1in,height=1.25in,clip,keepaspectratio]{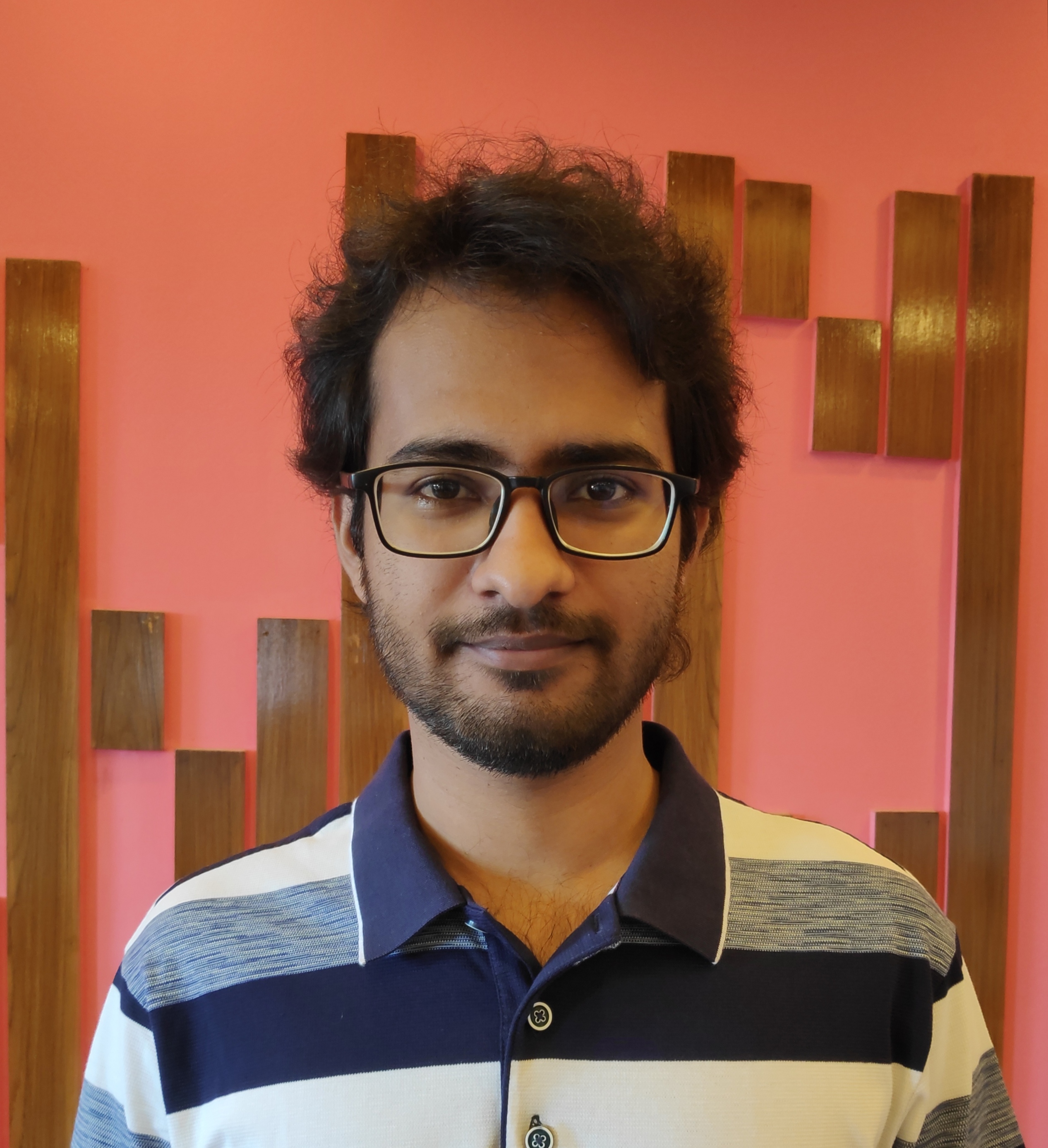}
\end{wrapfigure}\par
    \noindent \textbf{Sohan Salahuddin Mugdho} is an undergraduate student with the Department of Electrical and Electronic Engineering at Bangladesh University of Engineering and Technology (BUET), Dhaka, Bangladesh. He is majoring in Communications and Signal Processing. He has been awarded the Dean’s List Scholarship in one academic term. His research interests are in the areas of machine learning, signal/image processing, computer vision, natural language processing and centralized/distributed differentially private machine learning algorithms.\par
    
\begin{wrapfigure}{l}{25mm} 
    \includegraphics[width=1in,height=1.25in,clip,keepaspectratio]{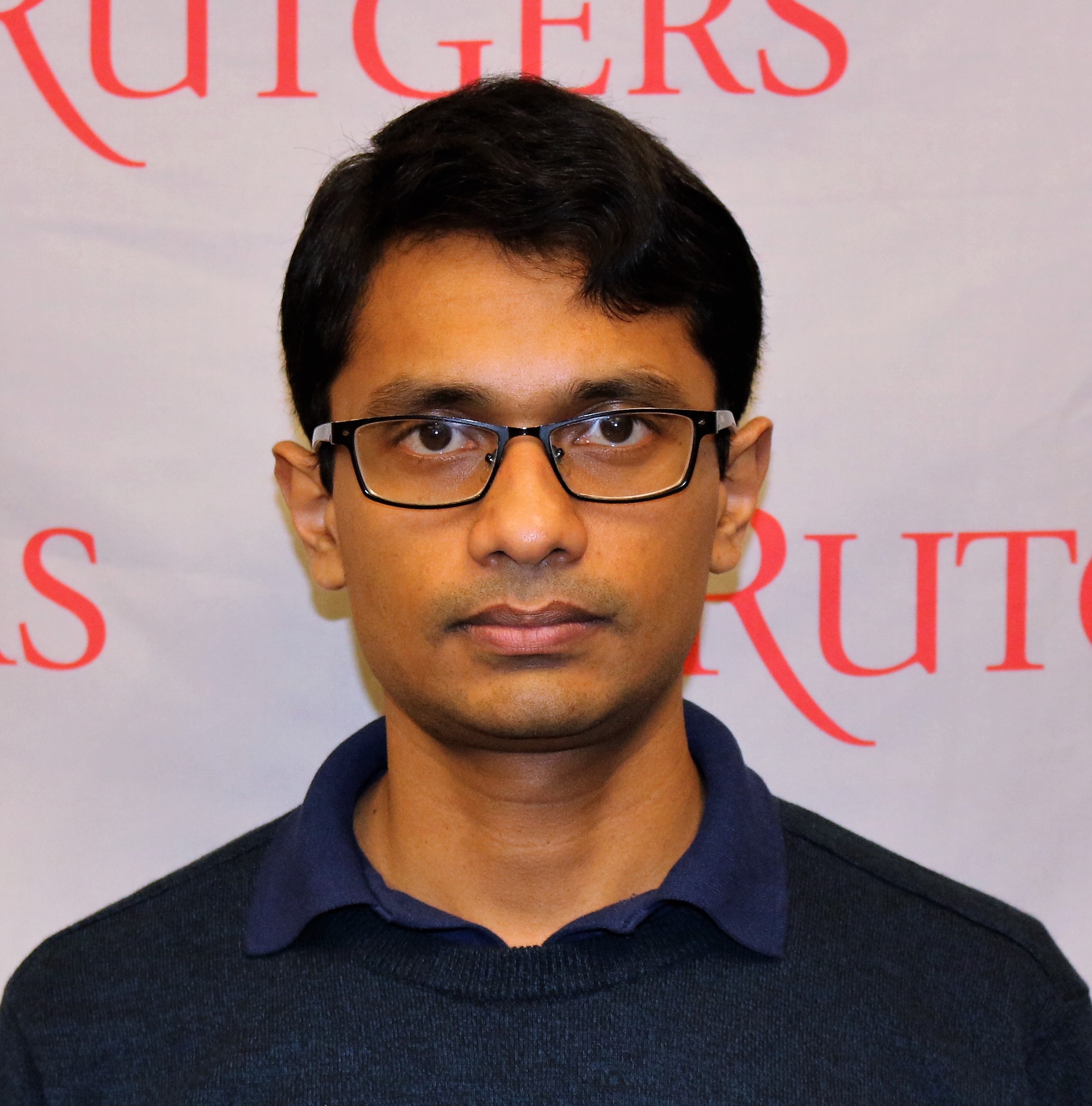}
\end{wrapfigure}\par
    \noindent \textbf{Hafiz Imtiaz} completed his PhD from Rutgers University, New Jersey, USA in 2020. He earned his second M.Sc. degree from Rutgers University in 2017, his first M.Sc. degree and his B.Sc. degree from Bangladesh University of Engineering and Technology (BUET), Dhaka, Bangladesh in 2011 and 2009, respectively. He is currently an Associate Professor with the Department of Electrical and Electronic Engineering at BUET. Previously, he worked as an intern at Qualcomm and Intel Labs, focusing on activity/image analysis and adversarial attacks on neural networks, respectively. His primary area of research includes developing privacy-preserving machine learning algorithms for decentralized data settings. More specifically, he focuses on matrix and tensor factorization, and optimization problems, which are core components of many modern machine learning algorithms.\par

\end{document}